\documentclass[journal,twoside,web]{ieeecolor}

\usepackage[usenames,dvipsnames]{xcolor}
\def\logoname{LOGO-generic-web}
\def\logoname{blank_logo}
\definecolor{subsectioncolor}{rgb}{0,0.541,0.855}

\usepackage[style=ieee]{biblatex}
\usepackage{./TLC}

\usepackage{circuitikz}
\usepackage{pgfplots}
\pgfplotsset{compat=1.16}
\usepackage{subcaption}
\usepgfplotslibrary{groupplots}
\usepackage{layouts}
\pgfplotscreateplotcyclelist{colors}{
        {CornflowerBlue},
        {BurntOrange},
        {RubineRed},
        {Emerald},
        {SpringGreen},
        {Goldenrod}
}

\usepackage{mathtools}
\usepackage{physics}
\usepackage{blox}
\usepackage{siunitx}

\title{Graphical Nonlinear System Analysis}

\author{Thomas Chaffey$^{1}$, Fulvio Forni$^{1}$ and Rodolphe Sepulchre$^{1}$
\thanks{The research leading to these results has received funding from the European
Research Council under the Advanced ERC Grant Agreement Switchlet n. 670645.}%
\thanks{$^{1}$University of Cambridge, Department of Engineering, Trumpington Street,
Cambridge CB2 1PZ, {\tt\small tlc37@cam.ac.uk}, {\tt\small \{f.forni, r.sepulchre\}@eng.cam.ac.uk}.}
}
\begin{document}
\maketitle

\begin{abstract}
        We use the recently introduced concept of a Scaled Relative Graph (SRG) to
        develop a graphical analysis of input-output properties of feedback systems.  The
        SRG of a nonlinear operator generalizes the Nyquist diagram of an LTI system.
        In the spirit of classical control theory, important robustness indicators of
        nonlinear feedback systems are measured as distances between SRGs.
\end{abstract}

\section{Introduction}

The graphical analysis of a feedback system via the Nyquist diagram of its return ratio 
is a foundation of classical control theory. It underlies the analysis concept of 
stability margins and the design concept of loop shaping, which themselves provide
the grounds for the gap metric \autocite{El-Sakkary1985a, Vinnicombe2000} and
$H_\infty$ control \autocite{Zames1981a}.

The Nyquist diagram also has a fundamental place in the theory of nonlinear systems
of the Lur'e form (that is, systems composed of an LTI forward path in feedback with a
static nonlinearity).  The circle and Popov criteria allow the stability of a Lur'e
system to be proved by verifying a geometric condition on the Nyquist diagram of the
LTI component \autocite{Desoer1975}.  The geometric condition is determined by the
properties of the static nonlinearity.  Notably, only the Nyquist diagram of the LTI
component is defined, owing to a lack of a suitable definition of phase for nonlinear
systems.  At best, the frequency response of a nonlinear system may be computed
approximately.  The describing function \autocite{Krylov1947, Blagquiere1966, Slotine1991} gives rise to a family
of Nyquist curves for a nonlinearity, parameterized by the amplitude of the input.
Other efforts to generalize frequency response to nonlinear systems include the work
of \textcite{Pavlov2006} on Bode diagrams for convergent systems, and the recently introduced
notions of nonlinear phase and singular angle by Chen \emph{et al.} \autocite{Chen2020, Chen2021a}.

In this paper, we show that the Scaled Relative Graph of \textcite{Ryu2021}
generalizes the Nyquist diagram of an LTI transfer function, and may be plotted for
nonlinear input/output operators.  The SRG has been introduced in the theory of
optimization to visualize incremental properties of nonlinear operators, that is,
properties that are measured between pairs of input/output trajectories, such as
Lipschitz continuity and maximal monotonicity.  Such properties may be verified by
checking geometric conditions on the SRG of an operator.  Algebraic manipulations to
the operator correspond to geometric manipulations to the SRG.  The SRG gives rise to
simple, intuitive and rigorous proofs of the convergence of many algorithms in convex
optimization.  Furthermore, the tool is particularly suited to proving tightness of
convergence bounds, and has been used to prove novel tightness results
\autocite{Ryu2021,Huang2020}.

The objective of this paper is to provide a bridge
between SRG analysis and the incremental input-output analysis
of feedback systems  \autocite{Desoer1975}.  Our main result
(Theorem~\ref{thm:robust})
establishes a generalization of the Nyquist theorem for stable nonlinear
operators. Based on the homotopy argument central to IQC analysis
\autocite{Megretski1997},
this result enables SRG
analysis to address well-posedness issues -- in contrast to \autocite{Ryu2021} for
example, where
well-posedness of operators is {\it assumed}. In the context of nonlinear
feedback system analysis, this
result
enables an elegant and classical definition of stability margins for nonlinear
feedback systems. We illustrate the generality of the approach by recovering
several classical results.

The second part of the paper aims at providing concrete illustrations
of the theory. As a first step, we show that the SRG of an LTI system is
derived from its Nyquist curve, and we also provide an analytical derivation
of the SRG of scalar-valued static nonlinearities. Preliminary results were presented
in the conference paper \autocite{Chaffey2021b}. We then illustrate the application of SRG analysis
to three representative examples from the literature. Our first example
(Section~\ref{sec:delay})
is a feedback loop involving delays and saturations, a classical benchmark of IQC
system analysis \autocite{Megretski1997}. Our SRG analysis provides an analytical
bound on the feedback gain which guarantees stability, and  closely matches previous
numerical bounds obtained by IQC analysis.
We stress however that the SRG analysis characterizes the incremental gain whereas
previous bounds were non-incremental. We furthermore give an analytical bound on
the incremental
$L_2$ gain of the closed loop.  This example illustrates the strength and
potential of SRG analysis for verifying {\it incremental} properties, which is of considerable
importance, even though it seems to have received little attention in IQC analysis
\autocite{Megretski1997, Jonsson2001, Wang2019}.
Our second example illustrates SRG analysis in so-called cyclic feedback systems
\autocite{Tyson1978, Mallet-Paret1996, Sontag2006, Arcak2006}. Here also SRG
analysis suggests a strong potential: on top of providing an elegant graphical interpretation
of existing results, we illustrate that SRG analysis provides analytical bounds on
the incremental gain of the feedback system and stability margins against
dynamical uncertainties. Our final example (Section~\ref{sec:congestion}) combines cascades
and delays in the analysis of a congestion control model previously studied in
\autocite{Summers2013}.
Here again, SRG analysis provides novel bounds on the incremental gain, and
generalizes the equilibrium-independent passivity analysis proposed in
\autocite{Summers2013}.

The rest of the paper is organised as follows.  We begin by introducing some
background material in Section~\ref{sec:background}, before defining the SRG in
Section~\ref{sec:srg}, and summarizing its main properties.
Section~\ref{sec:feedback} presents our main theoretical results on the SRG analysis
of feedback systems.  Section~\ref{sec:LTI} connects the SRG to the Nyquist diagram
of an LTI transfer function, and Section~\ref{sec:static} describes the SRGs of
important classes of static nonlinearities.  Three detailed examples are then
developed in Sections~\ref{sec:delay}, \ref{sec:cyclic} and~\ref{sec:congestion}.
Some concluding remarks and areas for future research are given in
Section~\ref{sec:conclusions}.

\section{Background and preliminaries}\label{sec:background}

\subsection{Signal spaces, operators and relations}

Let $\mathcal{L}$ denote a Hilbert space, equipped with an inner product,
$\bra{\cdot}\ket{\cdot}: \mathcal{L} \times \mathcal{L} \to \mathbb{C}$, and the
induced norm $\norm{x} \coloneqq \sqrt{\bra{x}\ket{x}}$.

We will pay particular attention to Lebesgue spaces of square-integrable functions.
Given a time axis, which for brevity we will always consider to be $\R_{\geq 0}$, and a field 
$\mathbb{F} \in \{\R, \C\}$, we define the space $L^n_2(\mathbb{F})$ by the set of
signals $u: \R_{\geq 0} \to \mathbb{F}^n$ such that
\begin{IEEEeqnarray*}{rCl}
        \norm{u} \coloneqq \left(\int_{0}^{\infty} \bar{u}(t) u (t) \dd{t}\right)^{\frac{1}{2}}
        < \infty,
\end{IEEEeqnarray*}
where $\bar{u}(t)$ denotes the conjugate transpose of $u(t)$.
The inner product of $u, y \in L_2^n(\mathbb{F})$ is defined by
\begin{IEEEeqnarray*}{rCl}
        \bra{u}\ket{y} \coloneqq \int_{0}^{\infty} \bar{u}(t) y (t) \dd{t}.
\end{IEEEeqnarray*}
The Fourier transform of $u \in L_2^n(\mathbb{F})$ is defined as
\begin{IEEEeqnarray*}{rCl}
        \hat{u}(j\omega) \coloneqq \int_0^\infty e^{-j\omega t}u(t) \dd{t}.
\end{IEEEeqnarray*}
We omit the dimension and field when they are immaterial or clear from context.

For some $T \in \R_{\geq 0}$, define the truncation operator $P_T$ by
\begin{equation*}
        (P_T u)(t) \coloneqq \left\{ \begin{array}{c c}
                        u(t) & t \leq T,\\
                        0 & t > T,
                \end{array}\right.
\end{equation*}
where $t \in \R_{\geq 0}$ and $u$ is an arbitrary signal.  Define the
\emph{extension of $L^n_2(\mathbb{F})$} \autocite{Zames1966}, \autocite[p.
22]{Willems1971}, \autocite[p. 172]{Desoer1975} to be the space
\begin{equation*}
        L^n_{2, e}(\mathbb{F}) \coloneqq \left\{ u: \R_{\geq 0} \to \mathbb{F}^n \; | \; \norm{P_T u} <
        \infty \text{ for all } T \in \R_{\geq 0} \right\}.
\end{equation*}

An \emph{operator}, or \emph{system}, on a space $\mathcal{X}$, 
is a possibly
multi-valued map $R: \mathcal{X} \to \mathcal{X}$.
The identity operator, which maps $u \in \mathcal{X}$ to itself, is denoted by $I$.
The \emph{graph}, or \emph{relation}, of an operator, is the set $\{u, y\; | \; u \in
\dom{R}, y \in R(u)\} \subseteq \mathcal{X}\times\mathcal{X}$.  We use the notions of an operator and its relation
interchangeably, and denote them in the same way.  The relation of an operator may be
thought of as an input/output partition of a behavior \autocite[Def. 3.3.1]{Polderman1997}.

The usual operations on functions can be extended to relations:
\begin{IEEEeqnarray*}{rCl}
        S^{-1} &=& \{ (y, u) \; | \; y \in S(u) \}\\
        S + R &=& \{ (x, y + z) \; | \; (x, y) \in S, (x, z) \in R \}\\
        SR &=& \{ (x, z) \; | \; \exists\; y \text{ s.t. } (x, y) \in R, (y, z) \in S \}.
\end{IEEEeqnarray*}
Note that $S^{-1}$ always exists, but is not an inverse in the usual sense.  In
particular, it is in general not the case that $S^{-1}S = I$.  If, however, $S$ is an
invertible function, the relational inverse and functional inverse coincide, so the
notation $S^{-1}$ can be used without ambiguity.

An operator $R$ on $L_2$ or $L_{2, e}$ is said to be \emph{causal} if $P_T(R(u)) =
P_T(R(P_T u))$ for all $u$.

\subsection{Incremental input/output analysis}\label{sec:input-output}

The input/output approach to nonlinear systems analysis originated in the
dissertation of George Zames \autocite{Zames1960} and early work by Irwin Sandberg
\autocite{Sandberg1965a}.  Noting that the amplification of a nonlinear system was,
in general, dependent on the input, Zames introduced the notion of the \emph{incremental gain}
of a system, which characterizes the worst case amplification a system is capable of. 
Incremental properties feature heavily in Desoer and Vidyasagar's classic text
\autocite{Desoer1975}.  The general pattern is that requiring a property to be
verified for every possible input, rather than just a single distinguished input
($u = 0$, for example), leads to much stronger results, often comparable to the
results that may be proved for linear systems.  This is perhaps unsurprising, as any
property of a linear system is automatically incremental.  The study of properties relative to the
zero input, however, have dominated nonlinear input/output theory since these early days.

We now define the input/output properties of systems considered in this paper.  
We begin with a definition of incremental stability.

\begin{definition}
        Let $R: L_2 \to L_2$.  The \emph{incremental $L_2$ gain} of $R$ is 
        \begin{IEEEeqnarray*}{rCl}
                \mu \coloneqq \sup_{u_1, u_2 \in \dom{R}} \frac{\norm{y_1 -
                                y_2}}{\norm{u_1 - u_2}}, 
        \end{IEEEeqnarray*}
        where $y_1 \in R(u_1)$, $y_2 \in R(u_2)$.  If $\mu < \infty$, $R$ is said
        to have \emph{finite incremental $L_2$ gain}, or be \emph{incrementally
        $L_2$ stable}.
\end{definition}

The second class of properties relate to passivity.

\begin{definition}\label{def:passive}
        Let $R: L_{2} \to L_{2}$. Then:
        \begin{enumerate}
                \item $R$ is said to be \emph{incrementally positive} if
                \begin{IEEEeqnarray*}{rCl}
                        \bra{u_1 - u_2}\ket{y_1 - y_2} \geq 0
                \end{IEEEeqnarray*}
                for all $u_1, u_2 \in \dom{R}$ and $y_1 \in R(u_1), y_2 \in R(u_2)$.
        \item $R$ is said to be \emph{$\lambda$-input-strict incrementally positive} if
                \begin{IEEEeqnarray*}{rCl}
                        \bra{u_1 - u_2}\ket{y_1 - y_2} \geq \lambda \norm{
                        u_1 - u_2}^2
                \end{IEEEeqnarray*}
                for all $T \geq 0$, all $u_1, u_2 \in \dom{R}$ and $y_1 \in R(u_1), y_2 \in R(u_2)$.
        \item $R$ is said to be \emph{$\gamma$-output-strict incrementally positive} if
                \begin{IEEEeqnarray*}{rCl}
                        \bra{u_1 - u_2}\ket{y_1 - y_2} \geq \gamma \norm{
                        y_1 - y_2}^2
                \end{IEEEeqnarray*}
                for all $u_1, u_2 \in \dom{R}$ and $y_1 \in R(u_1),
                y_2 \in R(u_2)$.\qedhere
        \end{enumerate}
\end{definition}

For causal operators on $L_2$, incremental positivity coincides with the stronger
property of incremental passivity (the proof follows the same lines as
\autocite[Lemma 2, p. 200]{Desoer1975}).  In the language of optimization, incremental
positivity is called \emph{monotonicity}.
Monotone operator theory originated in the study of networks of nonlinear resistors.
The prototypical monotone operator was Duffin's \emph{quasi-linear} resistor
\autocite{Duffin1946}, a resistor with increasing, but not necessarily linear, $i-v$
characteristic.  The modern notion of a monotone operator was introduced by
\textcite{Minty1960, Minty1961}. Following the influential paper of
\textcite{Rockafellar1976}, monotone operators have become
a cornerstone of optimization theory.  Monotone operator methods have seen a surge of
interest in the last decade, due to their applicability to large-scale and nonsmooth
problems \cite{Ryu2016, Ryu2021a, Combettes2011, Combettes2018}.  SRGs have been
developed to prove convergence of these optimization methods.

Monotone operator theory is closely related to the classical input/output theory of
nonlinear systems.  All of the properties studied in the theory of monotone
operators correspond to a property in input/output system theory, the difference
being that the former are defined for an arbitrary Hilbert space, while the latter
are defined on $L_2$ or $L_{2,e}$.  
Table~\ref{tab:dictionary} shows these equivalences.  

\def\arraystretch{1.5}
\begin{table*}[ht]
        \centering
        \begin{tabular}{c p{3cm} p{3cm} p{2cm}}
                property & $L_2$ & $L_{2,e}$ & Hilbert \\ \hline
                $\norm{y_1 - y_2} \leq \mu \norm{u_1 - u_2}$ & finite incremental
                gain & finite incremental gain$^*$ & Lipschitz \\
                $\bra{u_1 - u_2}\ket{y_1 - y_2} \geq 0$ & incremental positivity &
                incremental passivity$^*$ & monotonicity \\
                $\bra{u_1 - u_2}\ket{y_1 - y_2} \geq \lambda \norm{u_1 - u_2}^2$ &
                incremental input-strict positivity & incremental input-strict
                passivity$^*$ & strong monotonicity or coercivity \\
                $\bra{u_1 - u_2}\ket{y_1 - y_2} \geq \gamma \norm{y_1 - y_2}^2$ &
                incremental output-strict positivity & incremental output-strict
                passivity$^*$ & cocoercivity 
        \end{tabular}
        \caption{A partial bilingual dictionary from input/output system theory to
                monotone operator theory.  The first column gives properties between
                pairs of inputs $u_1, u_2$ and the corresponding outputs $y_1, y_2$.
                Greek letters denote positive scalars.
                The second and third columns give the system theory names of
properties of operators on either $L_2$ or $L_{2,e}$, as defined by
\autocite{Desoer1975}.  The names in the second column ($^*$) apply if the property
holds for every truncation $\bra{P_T u_1 - P_T u_2}\ket{P_T y_1 - P_T y_2}$, $T > 0$.  The fourth column gives the names of these properties in
monotone operator theory, for operators on an arbitrary Hilbert space - see, for
example, \autocite{Ryu2021}.}
        \label{tab:dictionary}
\end{table*}

\section{Scaled relative graphs}\label{sec:srg}
We define SRGs in the same way as \textcite{Ryu2021}, with the minor modification of
allowing complex valued inner products. 

Let $\mathcal{L}$ be a Hilbert space.
The angle between $u, y \in \mathcal{L}$ is defined as
\begin{IEEEeqnarray*}{rCl}
        \angle(u, y) \coloneqq \acos \frac{\Re \bra{u}\ket{y}}{\norm{u}\norm{y}} \in [0, \pi]. 
\end{IEEEeqnarray*}

Let $R: \mathcal{L} \to \mathcal{L}$ be an operator.  Given $u_1, u_2 \in
\mathcal{U} \subseteq \mathcal{L}$, $u_1 \neq u_2$, define the set of complex numbers $z_R(u_1, u_2)$ by
\begin{IEEEeqnarray*}{rCl}
        z_R(u_1, u_2) \coloneqq &&\left\{\frac{\norm{y_1 - y_2}}{\norm{u_1 - u_2}} e^{\pm j\angle(u_1 -
u_2, y_1 - y_2)}\right.\\&&\bigg|\; y_1 \in R(u_1), y_2 \in R(u_2) \bigg\}.
\end{IEEEeqnarray*}
If $u_1 = u_2$ and there are corresponding
outputs $y_1 \neq y_2$, then
$z_R(u_1, u_2)$ is defined to be $\{\infty\}$.  If $R$ is single valued at $u_1$,
$z_R(u_1, u_1)$ is the empty set.

The \emph{Scaled Relative Graph} (SRG) of $R$ over $\mathcal{U} \subseteq \mathcal{L}$ is then given by
\begin{IEEEeqnarray*}{rCl}
        \srg[\mathcal{U}]{R} \coloneqq \bigcup_{u_1, u_2 \in\, \mathcal{U}}  z_R(u_1, u_2).
\end{IEEEeqnarray*}
If $\mathcal{U} = \mathcal{L}$, we write $\srg{R} \coloneqq \srg[\mathcal{L}]{R}$.

If $R$ is linear and $\dom{R}$ is a linear subspace of $\mathcal{L}$, $Ru_1 -
Ru_2 =
R(u_1 - u_2) = Rv$ for some $v \in \dom{R}$,
and we can define
\begin{IEEEeqnarray*}{rCl}
        z_R(v) \coloneqq \frac{\norm{R v}}{\norm{v}} e^{\pm j\angle(v, Rv)}
\end{IEEEeqnarray*}
and
\begin{IEEEeqnarray*}{rCl}
        \srg[\dom{R}]{R} \coloneqq \left\{ z_R(v) \middle| v \in
        \dom{R}, v \neq 0 \right\}.
\end{IEEEeqnarray*}
In the special case that $R$ is linear and time invariant with transfer function
$R(s)$, and $v(t) = e^{j\omega t}$,
$\lim_{T \to \infty}\norm{R(P_T v)}/\norm{P_T v} = |R(j\omega)|$ and
$\lim_{T\to\infty}\angle(P_T v, R(P_T v))$ = $|\arg{R(j\omega)}|$ (where
$\arg{R(j\omega)}$ is measured between $-\pi$ and $\pi$).  Thus
the gain and phase of the SRG generalize the classical notions of the gain and phase
of an LTI transfer function.  

\subsection{System properties from SRGs}
If $\mathcal{A}$ is a class of operators, we define the SRG of $\mathcal{A}$ by
\begin{IEEEeqnarray*}{rCl}
        \srg{\mathcal{A}} \coloneqq \bigcup_{R \in \mathcal{A}} \srg{R}.
\end{IEEEeqnarray*}
Note that a class of operators can be a single operator.

A class $\mathcal{A}$, or its SRG, is called \emph{SRG-full} if
\begin{IEEEeqnarray*}{rCl}
        R \in \mathcal{A}\quad \iff \quad \srg{R} \subseteq \srg{\mathcal{A}}.
\end{IEEEeqnarray*}
By construction, the implication $R \in \mathcal{A} \implies \srg{R} \subseteq
\srg{\mathcal{A}}$ is true.  The value of SRG-fullness is in the reverse implication: 
$\srg{R} \subseteq \srg{\mathcal{A}} \implies R \in \mathcal{A}$.  This allows class
membership to be tested graphically.  If $\mathcal{A}$ is the class of systems with a
particular system property, SRG-fullness of $\mathcal{A}$ allows this property to be
verified for a particular operator $R$ by plotting its SRG.  If $\srg{R} \subseteq
\srg{\mathcal{A}}$, $R$ has the property.

The following proposition gives the SRGs of the classical system properties introduced in
Section~\ref{sec:input-output}.  
\begin{proposition}\label{prop:properties}
        The SRGs of incrementally positive systems (top left),
        input-strict incrementally positive systems (top right), output-strict
        incrementally positive systems (bottom right) and incrementally $L_2$ bounded
        systems (bottom left) are shown below.  
        \begin{center}
        \includegraphics{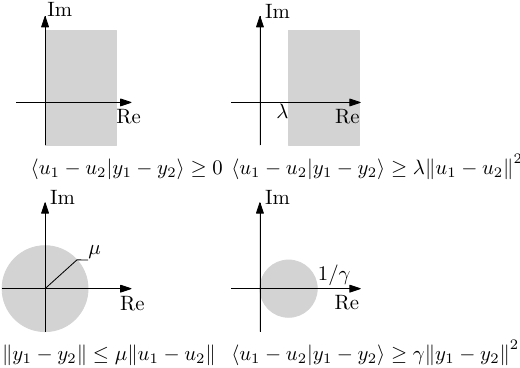}
        \end{center}
        All of these classes are SRG-full.
\end{proposition}

\begin{proof}
These SRGs are proved in
\autocite{Ryu2021}, and all of the shapes follow from quick calculations.
SRG-fullness follows from \autocite[Thm. 3.5]{Ryu2021}.
\end{proof}

SRG-fullness of the classes of Proposition~\ref{prop:properties} means, for example,
that if the SRG of a system lies in the right half plane, the system is incrementally
positive, and if the SRG of a system is bounded, the system has
finite incremental $L_2$ gain.  These are reminiscent of the facts
that an LTI system is passive if its Nyquist diagram lies in the right half plane,
and has finite $H_\infty$ norm if its Nyquist diagram is bounded.  We will show in
Section~\ref{sec:LTI} that Proposition~\ref{prop:properties} is indeed a
generalization of these classical facts.

The properties of finite incremental $L_2$ gain and incremental positivity are particular
examples of incremental Integral Quadratic Constraints (IQCs)
\autocite{Megretski1997}.
A striking corollary of \textcite[Thm. 3.5]{Ryu2021} is that any SRG
defined by a static incremental IQC is SRG-full.

\begin{proposition}
        Let $u_i(t)$ denote the input to an arbitrary operator on $L_2$, and $y_i(t)$
        denote a corresponding output.  Let $\Delta u = u_1 - u_2$ and $\Delta y =
        y_1 - y_2$, and $\hat{x}(\omega)$ denote the Fourier transform of signal
        $x(t)$.  Then the classes of operators which obey either of the constraints
        \begin{IEEEeqnarray}{rCl}
                \int^\infty_{-\infty}\begin{pmatrix} \Delta\hat{u}(\omega)\\ \Delta
                \hat{y}(\omega) \end{pmatrix}\tran
                \begin{pmatrix} a & b\\ c & d \end{pmatrix} \begin{pmatrix} \Delta
        \hat{u}(\omega)\\ \Delta \hat{y}(\omega) \end{pmatrix} \dd{\omega} &\geq& 0,\label{eq:quad_1}\\
        \int^\infty_{0} \begin{pmatrix} \Delta u(t)\\ \Delta y(t) \end{pmatrix}\tran
        \begin{pmatrix} a & b\\ c & d \end{pmatrix} \begin{pmatrix} \Delta u(t)\\
\Delta y(t) \end{pmatrix}\dd{t} &\geq& 0,\label{eq:quad_2}
        \end{IEEEeqnarray}
        where $a, b, c, d \in \R$, are SRG-full. 
\end{proposition}

\begin{proof}
       Equation~\eqref{eq:quad_1} gives 
       \begin{IEEEeqnarray*}{rCl}
               a\norm{\Delta \hat{u}}^2 + (b + c)\bra{\Delta \hat{u}}\ket{\Delta
               \hat{y}} + d \norm{\Delta \hat{y}}^2 \geq 0.
       \end{IEEEeqnarray*}
       By Parseval's theorem, this is equivalent to
       \begin{IEEEeqnarray*}{rCl}
               a\norm{\Delta u}^2 + (b + c)\bra{\Delta u}\ket{\Delta y} + d
               \norm{\Delta y}^2 \geq 0,
       \end{IEEEeqnarray*}
       which is also implied by \eqref{eq:quad_2}.  The result then follows from
       \autocite[Thm. 3.5]{Ryu2021}.
\end{proof}

A class of operators defined defined by a geometric region
is SRG-full.  
\begin{proposition}
        Let $\mathcal{C} \subseteq \mathbb{C}$.  The class of operators $\mathcal{A}$
        defined by
        \begin{IEEEeqnarray*}{rCl}
                \mathcal{A} \coloneqq \{R \mbox{ an operator } | \srg{R} \subseteq
                \mathcal{C}\}
        \end{IEEEeqnarray*}
        is SRG-full.
\end{proposition}
\begin{proof}
        The definition of $\mathcal{A}$ can be written as
        \begin{IEEEeqnarray*}{rCl}
                R \in \mathcal{A} \iff \srg{R} \subseteq \mathcal{C},
        \end{IEEEeqnarray*}
        which is the definition of SRG-fullness.
\end{proof}

This fact is particularly useful for system analysis, as it allows the
SRG of an operator to be over-approximated by a geometric region if, for example, the
precise SRG is unknown, or the SRG does not obey the properties necessary to apply a
theorem.  Over-approximating an SRG simply amounts to making the analysis more
conservative.

\subsection{Interconnections}

Under mild conditions on the SRG, system interconnections correspond to geometric
manipulations of their SRGs.  These interconnection results are proved by
\textcite{Ryu2021} in Theorems 4.1-4.5.  We recall the statements of these theorems
in the following five propositions.

\begin{proposition}
        If $\mathcal{A}$ and $\mathcal{B}$ are SRG-full, then
        $\mathcal{A}\cap\mathcal{B}$ is SRG-full, and
        \begin{IEEEeqnarray*}{rCl}
                \srg{\mathcal{A}\cap\mathcal{B}} = \srg{\mathcal{A}}
                \cap\srg{\mathcal{B}}.
        \end{IEEEeqnarray*}
\end{proposition}

\begin{proposition}
        Let $\alpha \in \R, \alpha \neq 0$. If $\mathcal{A}$ is a class of operators,
        \begin{IEEEeqnarray*}{rCl}
                \srg{\alpha\mathcal{A}} &=& \srg{\mathcal{A}\alpha} =
                \alpha\srg{\mathcal{A}},\\
                \srg{I + \mathcal{A}} &=& 1 + \srg{\mathcal{A}}.
        \end{IEEEeqnarray*}
        Furthermore, if $\mathcal{A}$ is SRG-full, then $\alpha\mathcal{A}$, $\mathcal{A}\alpha$
        and $I + \mathcal{A}$ are SRG-full.
\end{proposition}

We define inversion in the
complex plane by the M\"obius transformation $re^{j\omega} \mapsto (1/r)e^{j\omega}$.
This is ``inversion in the unit circle'': points outside the unit circle map to the
inside, and vice versa.  The points $0$ and $\infty$ are exchanged under inversion.
\begin{proposition}\label{prop:inversion}
        If $\mathcal{A}$ is a class of operators, then
        \begin{IEEEeqnarray*}{rCl}
                \srg{\mathcal{A}^{-1}} = (\srg{\mathcal{A}})^{-1}.
        \end{IEEEeqnarray*}
        Furthermore, if $\mathcal{A}$ is SRG-full, then $\mathcal{A}^{-1}$ is
        SRG-full.
\end{proposition}

Define the line segment between $z_1, z_2 \in \C$ as $[z_1, z_2] \coloneqq \{\alpha
                                z_1 + (1 - \alpha) z_2\, |\, \alpha \in [0, 1]\}$.
A class of operators $\mathcal{A}$ is said to
satisfy the \emph{chord property} if $z \in \srg{\mathcal{A}}\setminus\{\infty\}$ implies $[z,
\bar z] \subseteq \srg{\mathcal{A}}$.

\begin{proposition}\label{prop:summation}
        Let $\mathcal{A}$ and $\mathcal{B}$ be classes of operators, such that $\infty \nin \srg{\mathcal{A}}$ and $\infty \nin
        \srg{\mathcal{B}}$. Then:
        \begin{enumerate}
                \item if $\mathcal{A}$ and $\mathcal{B}$ are SRG-full, then
                        $\srg{\mathcal{A} + \mathcal{B}} \supseteq \srg{\mathcal{A}} +
                        \srg{\mathcal{B}}$.
                \item if either $\mathcal{A}$ or $\mathcal{B}$ satisfies the chord
                        property, then 
                        $\srg{\mathcal{A} + \mathcal{B}} \subseteq \srg{\mathcal{A}} +
                        \srg{\mathcal{B}}$.
        \end{enumerate}
\end{proposition}
Under additional assumptions, $\infty$ can be allowed - see the discussion following \autocite[Thm.
4.4]{Ryu2021}.

Define the \emph{right-hand arc}, $\rarc{z, \bar{z}}$, between $z$ and $\bar{z}$ to be the arc between $z$
and $\bar{z}$ with centre on the origin and real part greater than or equal to
$\Re(z)$. The \emph{left-hand arc}, $\larc{z, \bar{z}}$, is defined the same way, but with real part less
than or equal to $\Re(z)$.  Formally,
\begin{IEEEeqnarray*}{rCl}
        \rarc{z, \bar{z}} &\coloneqq& \large\{re^{j(1-2\theta)\phi} \;\large|\; z =
        re^{j\phi},\\&& \phi \in (-\pi, \pi],\, \theta \in [0, 1],\, r\geq 0\large\},\\
                \larc{z, \bar{z}} &\coloneqq& -\rarc{-z, -\bar{z}}.
\end{IEEEeqnarray*}
A class of operators $\mathcal{A}$ is said to satisfy the \emph{right hand (resp. left hand) arc property} if,
for all $z \in \srg{\mathcal{A}}$, $\rarc{z, \bar{z}} \in \srg{\mathcal{A}}$  (resp.
$\larc{z,
\bar{z}} \in \srg{\mathcal{A}}$).

\begin{proposition}\label{prop:composition}
        Let $\mathcal{A}$ and $\mathcal{B}$ be classes of operators, such that
        $\infty \nin \srg{\mathcal{A}}$, $\mathcal{A} \neq \varnothing$, $\infty \nin
        \srg{\mathcal{B}}$ and $\mathcal{B} \neq \varnothing$.  Then:
        \begin{enumerate}
                \item if $\mathcal{A}$ and $\mathcal{B}$ are SRG-full, then
                        $\srg{\mathcal{A}\mathcal{B}} \supseteq \srg{\mathcal{A}}
                        \srg{\mathcal{B}}$.
                \item if either $\mathcal{A}$ or $\mathcal{B}$ satisfies an arc 
                        property, then 
                        $\srg{\mathcal{A}\mathcal{B}} \subseteq \srg{\mathcal{A}}
                        \srg{\mathcal{B}}$.
        \end{enumerate}
\end{proposition}
Under additional assumptions, $\infty$ and $\varnothing$ can be allowed -- see the discussion following \autocite[Thm.
4.5]{Ryu2021}.

\subsection{Scaled graphs about particular solutions}\label{sec:input_dependent}

Scaled relative graphs capture the behavior of an operator with respect to any
possible operating point.  However, the behavior about one or several specific
inputs (for example, stable equilibria) may be of particular interest.
The methods of this paper apply equally to the analysis of properties with respect to
particular inputs, via the \emph{scaled graph} (SG). For notational
convenience, we only define the SG over the full space, but the SG can be
restricted to a subset of the input space in the same way as the SRG.

\begin{definition}
        Let $R: \mathcal{L} \to \mathcal{L}$.
        The \emph{scaled graph} of $R$ over $\mathcal{L}$ with
        respect to the input $u^\star$ is
        \begin{IEEEeqnarray*}{rCl}
                \sg[u^\star]{R} \coloneqq \bigcup_{u \in \mathcal{L}} z_R(u,
                u^\star).
        \end{IEEEeqnarray*}
\end{definition}

Note that the SG of an LTI operator with respect to any input is equal to its SRG.

The graphical algebra of SRGs applies to SGs with very little modification - the only
requirement is that interconnected SGs are defined with respect to compatible inputs.
 In the remainder of this
section, we highlight the difference between incremental and
input-specific properties, using the example of positivity.

\begin{definition}
        An operator $H: L_2 \to L_2$ is said to be \emph{positive about
        $u^\star \in L_2$} if, for all $u \in L_2$, $y \in H(u)$ and
        $y^\star \in H(u^\star)$, $\bra{u - u^\star}\ket{y - y^\star} > 0$.
\end{definition}

From this definition, it follows immediately that an operator is positive about
$u^\star$ if and only if its SG about $u^\star$ belongs to the closed right half
plane.  However, this does not mean its SG about any other input necessarily lies in
the right half plane -
Figure~\ref{fig:negative_resistor} gives such an example.

\begin{figure}[h]
        \centering
        \includegraphics{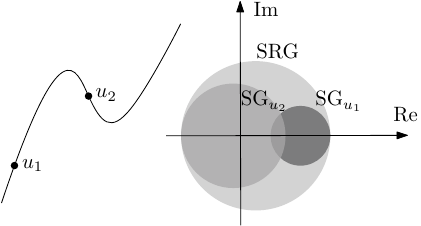}
        \caption{On the left is the $i-v$ characteristic of a resistor with a region of negative
        resistance.  The resistor is positive about some operating points (including
$u_1$), but not about others
(such as $u_2$).  On the right are the SGs computed at $u_1$ and $u_2$, as well as
the SRG.}
       \label{fig:negative_resistor}
\end{figure}

Taking the union of SGs over multiple trajectories allows properties that lie between
trajectory-dependent and incremental to be verified.  For example,
\textcite{Hines2011} define \emph{equilibrium-independent passivity} to be passivity
with respect to every constant input to the system (under assumptions on the system
that ensure that there is a constant output for every constant input).  This can be verified by
checking that the union of SGs over constant inputs lies in the right half plane.

\section{Feedback analysis with scaled relative graphs}\label{sec:feedback}

In this section, we demonstrate the use of scaled relative graphs for the stability analysis of
feedback interconnections.  We begin by using the SRG to generalize the Nyquist
criterion to a stable nonlinear operator in unity gain negative feedback, and
introduce a nonlinear stability margin.  We then formulate a general stability
theorem by inflating the point $-1$ to the negative of the SRG of an operator in the
feedback path, and show that this theorem encompasses both the incremental small gain
and incremental passivity theorems.

This stability theorem relies on viewing a feedback interconnection as the inverse of
a parallel interconnection.  The conditions of the theorem ensure that the parallel
interconnection has a strictly positive lower bound on its incremental gain; it then
follows that its inverse has an upper bound on its incremental gain.
In order to show that a feedback interconnection leads to a well-defined operator,
we use a homotopy argument
similar to \textcite{Megretski1997}.  We place a gain $\tau \in [0, 1]$ in the
feedback loop, and assume stability for $\tau = 0$ (no feedback).
We then use SRGs to show stability for every $\tau \in (0, 1]$, which implies that
there is no loss of stability as the feedback is introduced.  This allows us to use
SRG analysis to prove not only stability, but also well-posedness.

\subsection{A Nyquist stability criterion for stable nonlinear
operators}\label{subsec:Nyquist}

The Nyquist criterion characterizes the stability of a transfer function $L$ in
unity gain negative feedback (Figure~\ref{fig:nyq_fb}) in terms of the distance between
the Nyquist diagram of $L$ and the point -1.  This distance is
called the \emph{stability margin}, and is the inverse of the $H_\infty$ norm of the
sensitivity transfer function \autocite[p. 50]{Doyle1992}.  In this section, we show that the
Nyquist criterion and stability margin can be generalized to stable nonlinear operators by
replacing the Nyquist diagram with an SRG.  For such stable nonlinear operators, the closed loop system is stable if the
SRG of the loop operator leaves $-1$ on the left.

\begin{figure}[ht]
        \centering
        \includegraphics{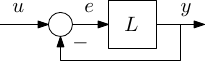}
        \caption{Unity gain negative feedback around the relation $L$.}%
        \label{fig:nyq_fb}
\end{figure}

\begin{theorem}\label{thm:Nyquist}
        Let $L: L_2 \to L_2$ be an operator with finite incremental $L_2$ gain. The
        closed loop operator shown in Figure~\ref{fig:nyq_fb} maps $L_2$ to $L_2$ and has finite incremental $L_2$ gain
        from $u$ to $y$ if
        \begin{IEEEeqnarray}{rCl}
                \label{eq:nyquist}
                0 &\nin& 1 + \tau\srg{L} \quad \mbox{for all }\tau \in (0, 1].
        \end{IEEEeqnarray}
        The closed loop gain from $u$ to $e$ in Figure~\ref{fig:nyq_fb} is bounded above by
        $1/s_m$, where $s_m$ is the shortest distance between $\srg{L}$ and $-1$.
\end{theorem}

\begin{proof}
We place a gain of $\tau$ in the feedback path, and show that the mapping from $\tau$ to the incremental gain from $u$ to $y$ is
continuous if \eqref{eq:nyquist} holds.  The operator from $u$ to $y$ is given by $(L^{-1} + \tau I)^{-1}$. Let the
distance between $\srg{L^{-1}}$ and $-\tau$ be $r_\tau > 0$.  Then $\srg{L^{-1} +
\tau I}$ is at least a distance of $r_\tau$ from the origin, so its inverse is at most
$r_\tau$ from the origin, giving a bound of $1/r_\tau$ on the
incremental gain from $u$ to $y$, as illustrated below.  Condition~\eqref{eq:nyquist}
guarantees that $r_\tau > 0$ for all $\tau \in (0, 1]$.

        \begin{center}
                \includegraphics{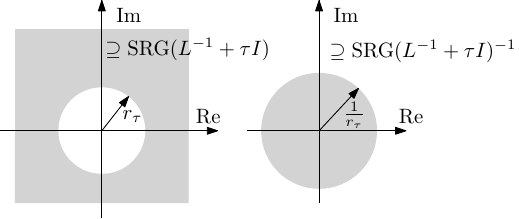}  
        \end{center}

Let $\epsilon > 0$ be smaller than $r_\tau$.  Then there exists a $\delta$ (positive or negative) such
that, if $\tau$ is changed to $\tau +
\delta$, the distance $r_\tau$ decreases by $\epsilon$.  Furthermore, as $\epsilon \to 0$, $\delta
\to 0$. This is a statement of the fact that the distance between a set
and a point varies continuously with the position of the point.  The closed
loop incremental gain bound then increases to $1/(r_\tau -
\epsilon)$. This
is bounded provided $\epsilon < r_\tau$ (in which case $\delta$ small
enough that $-(\tau + \delta)$ doesn't intersect $\srg{L^{-1}}$) and
approaches $r_\tau$ as $\delta \to 0$.  This shows continuity from $\tau$ to the
closed loop incremental gain from $u$ to $y$, and shows that finite incremental gain is preserved
provided $\srg{L^{-1}}$ never intersects $-1/\tau$.  In particular, all inputs in
$L_2$ continue to map to outputs in $L_2$.  We conclude finite incremental gain from
$u$ to $y$ by setting $\tau = 1$.

To prove the second part of the theorem, note that the relation from $u$ to $e$ is
given by
\begin{IEEEeqnarray*}{rCl}
        e = (I + L)^{-1} u.
\end{IEEEeqnarray*}
If $\srg{I + L}$ is bounded away from 0 by a distance $s_m$, then $(I +
L)^{-1}$ has an $L_2$ gain bound of $1/s_m$.  
\end{proof}

\subsection{A general feedback stability theorem}\label{subsec:robust}

The Nyquist stability criterion presented in the previous section can be generalized
to allow a second nonlinear operator in the feedback path, by
inflating the point $-1$ into the SRG of the feedback operator.  The following
theorem encompasses the classical small gain and passivity theorems as special cases.  

\begin{figure}[ht]
        \centering
        \includegraphics{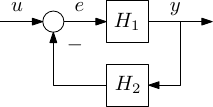}
        \caption{Negative feedback interconnection of $H_1$ and $H_2$.}%
        \label{fig:sym_fb}
\end{figure}

Let $\mathcal{H}$ be a class of operators. By $\bar{\mathcal{H}}$, we will denote a
class of operators such that $\mathcal{H} \subseteq \bar{\mathcal{H}}$ and
$\srg{\bar{\mathcal{H}}}$ satisfies the chord property.

\begin{theorem}\label{thm:robust}
        Consider the feedback interconnection shown in Figure~\ref{fig:sym_fb}
        between any pair of operators $H_1 \in \mathcal{H}_1$ and $H_2 \in
        \mathcal{H}_2$, where $\mathcal{H}_1$ is a class of operators on $L_2$ with
        finite incremental gain, and
        $\mathcal{H}_2$ is a
        class of operators on $L_{2}$.  If, for all $\tau \in (0, 1]$,
        \begin{IEEEeqnarray*}{rCl}
                \srg{\mathcal{H}_1}^{-1}\cap -\tau \srg{\bar{\mathcal{H}}_2} =
                \varnothing,
        \end{IEEEeqnarray*}
        then the feedback interconnection maps $L_2$ to $L_2$ and has an incremental
        $L_2$ gain bound from $u$ to $y$ of $1/r_m$, where 
        $r_m$ is the shortest distance between
$\srg{\mathcal{H}_1^{-1}}$ and $-\srg{\bar{\mathcal{H}}_2}$.
\end{theorem}

The choice of which SRG to over-approximate is arbitrary.
In the theorem, we have chosen $\srg{\mathcal{H}_2}$, but it could just as well be
$\srg{\mathcal{H}_1^{-1}}$.

\begin{proof}[Proof of Theorem~\ref{thm:robust}]
        For a gain of $\tau$ in the feedback path, the class of operators from $u$ to $y$ is given by
        \begin{IEEEeqnarray*}{rCl}
                (\mathcal{H}_1^{-1} + \tau \mathcal{H}_2)^{-1}.
        \end{IEEEeqnarray*}
Suppose there exists a positive number $r_{\tau}$ such
        that $|z - w| \geq r_{\tau}$ for all $z \in \srg{\mathcal{H}_1^{-1}}$, $w\in \srg{-\tau
        \bar{\mathcal{H}}_2}$.
        \begin{center}
                \includegraphics{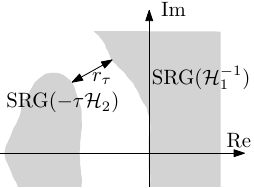}
        \end{center}
        Since $\srg{\mathcal{H_1}^{-1} + \tau H_2} \subseteq
\srg{\mathcal{H_1}^{-1}} + \tau \srg{\bar{\mathcal{H}}_2}$, where $H_2 \in
        \bar{\mathcal{H}}_2$, 
        it follows that $\srg{\mathcal{H_1}_1^{-1} + \tau H_2}$ is bounded away from zero by a distance of $r_{\tau}$
for all $H_2 \in \bar{\mathcal{H}}_2$. In particular,
        this holds for every operator $H_2 \in \mathcal{H}_2$.  
        \begin{center}
                \includegraphics{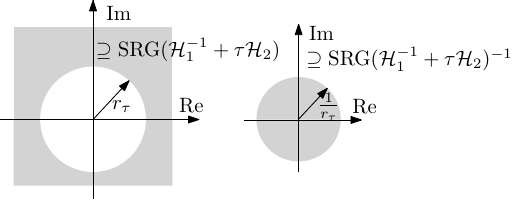}
        \end{center}
        Applying the inverse transformation
gives an incremental $L_2$ gain bound of $1/r_{\tau}$.

          Ensuring this holds for all $\tau \in
          (0, 1]$ means the finite incremental gain of $\mathcal{H}_1$ is never lost, so the
        feedback interconnection remains defined on $L_2$. $r_m$ corresponds to $r_1$.
\end{proof}

One case where the criteria of Theorem~\ref{thm:robust} are automatically satisfied
is the classical small gain setting.

\begin{corollary}\label{thm:small_gain}
        Consider the feedback interconnection shown in Figure~\ref{fig:sym_fb}
        between any pair of operators $H_1 \in \mathcal{H}_1$ and $H_2 \in
        \mathcal{H}_2$, where $\mathcal{H}_1$ and $\mathcal{H}_2$ are the classes of
        operators on $L_2$ with finite incremental $L_2$ gain bounds of $\gamma$ and
        $\lambda$, respectively.  If
                $\gamma\lambda < 1$,
        then the feedback interconnection maps $L_2$ to $L_2$ and has an incremental
        $L_2$ gain bound from $u$ to $y$ of
        $\gamma/(1-\gamma\lambda)$.
\end{corollary}

\begin{proof}
        The result follows directly from Theorem~\ref{thm:robust}.  The conditions of
        the theorem are shown to be satisfied by the geometry below.
        \begin{center}
                \includegraphics{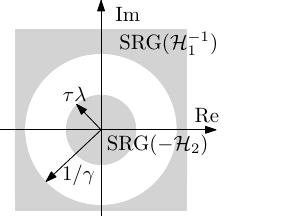}
        \end{center}
        \vspace{-10mm}
\end{proof}

The second case where the conditions of Theorem~\ref{thm:robust} are automatically
satisfied is in the feedback interconnection of incrementally positive systems.  The classical
incremental passivity theorem \autocite{Zames1966} is proved in the following
corollary.

\begin{corollary}\label{thm:passivity}
        Consider the feedback interconnection shown in Figure~\ref{fig:sym_fb}
        between any pair of operators $H_1 \in \mathcal{H}_1$ and $H_2 \in
        \mathcal{H}_2$, where $\mathcal{H}_1$  is the class of $\lambda$-input-strict
        incrementally positive operators which have an incremental $L_2$ gain bound of
        $\mu$, and $\mathcal{H}_2$ is the class of
        incrementally positive operators. Assume $\lambda > 0$.  Then the feedback
        interconnection maps $L_2$ to $L_2$ and has an incremental $L_2$ gain bound from $u$ to
        $y$ of $\mu^2/\lambda$.
\end{corollary}

\begin{proof}
        The SRGs of $H_1$ and $H_2$ are contained in the SRGs shown below.  Note that these both satisfy the chord
        property.
        \begin{center}
                \includegraphics{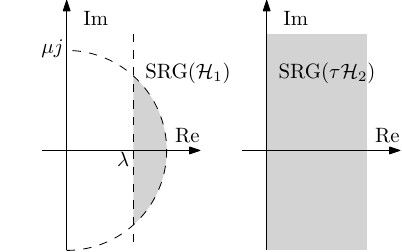}
        \end{center}

        The SRG of the inverse of the class of $\lambda$-input-strict
        incrementally positive operators is the circle with centre $1/(2\lambda)$ and
        radius $1/(2\lambda)$ (Proposition~\ref{prop:properties}).
        This circle is parameterized as
        $\{(1/\lambda)\cos(\theta) \exp(j\theta),\;|\; 0 \leq \theta \leq 2\pi\}$.
        The
        semicircle with centre at the origin, positive real part and radius $\mu$,
        which is the SRG of the class of incrementally positive operators with an
        incremental $L_2$ gain bound of $\mu$,
        is parameterized as $\{\mu\exp(j\phi),\;|\; -\pi/2 \leq \phi
        \leq \pi/2\}$.
        The result then follows from Theorem~\ref{thm:robust} and the geometry 
        below.

        \begin{center}
                \includegraphics{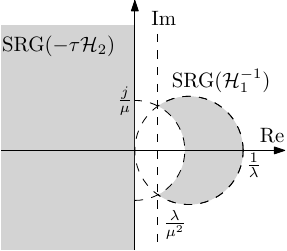}
        \end{center}
        \vspace{-10mm}
\end{proof}

Corollary~\ref{thm:passivity} characterizes the incremental gain of the closed loop.
We can also characterize the incremental positivity of the closed loop, with another
form of the classical passivity theorem.  The following theorem generalizes 
\autocite[Prop. 8]{Chaffey2021b}.

\begin{theorem}\label{thm:weak_passivity}
        Consider the feedback interconnection shown in Figure~\ref{fig:sym_fb}
        between any pair of operators $H_1 \in \mathcal{H}_1$ and $H_2 \in
        \mathcal{H}_2$, where $\mathcal{H}_1$ is the class of operators which are
        $\gamma$-output-strict incrementally positive, and $\mathcal{H}_2$ is the
        class of operators which are $\lambda$-input-strict incrementally positive.
        If
        \begin{IEEEeqnarray*}{rCl}
                \lambda + \gamma \geq 0,
        \end{IEEEeqnarray*}
        then the operator from $u$ to $y$ is $(\gamma + \lambda)$-output-strict
        incrementally positive.
\end{theorem}

\begin{proof}
        Assume, without loss of generality, that $\lambda < 0$.
        We first prove the case where $\lambda + \gamma > 0$.  This follows from the
        geometry shown below.
        \begin{center}
                \includegraphics{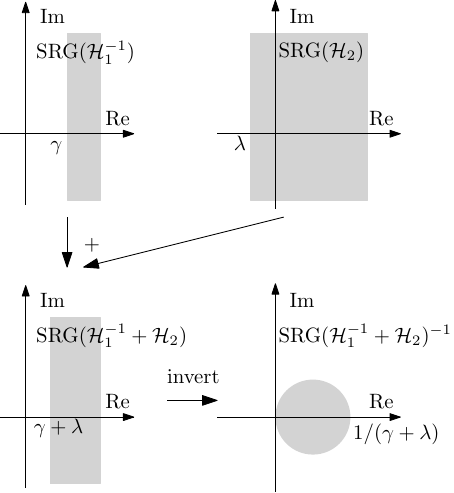}
        \end{center}

        The case where $\lambda + \gamma = 0$ then follows by taking the limit $\lambda
        \to -\gamma$, and allowing the radius of the circle in the final panel above
        to tend to $\infty$.
\end{proof}

The definition of a stability margin for nonlinear operators leads us naturally to
pose an ``$H_\infty$ design problem'', in the same vein as \textcite{Zames1981a},
to do with the maximization of the stability margin over a set of admissible
controllers.  A generalization of the $H_\infty$ design question to nonlinear
operators is as follows: given a plant $G$ (modelled by an operator on
$L_2$) in feedback with an uncertain block $\Delta$ known to be bounded by a particular
SRG, design a controller $C$ to maximize the distance between $\srg{C G}^{-1}$ and
$-\srg{\Delta}$.

\section{The scaled relative graph of an LTI transfer function}\label{sec:LTI}

In this section, we show that the SRG of a stable LTI transfer function is the convex
hull of its Nyquist diagram, under the Beltrami-Klein mapping.  We first presented
this result in \autocite{Chaffey2021b}, and it was noted by \textcite{Pates2021} that
this is a special case of a more general phenomenon involving the numerical range of
a linear operator.  This allows computational methods for the numerical range to be
applied directly to computation of the boundary of an SRG.

We begin by introducing some preliminaries from hyperbolic geometry in
Section~\ref{sec:hgeo}, before giving the main result in Section~\ref{subsec:lti}.

\subsection{Hyperbolic geometry}\label{sec:hgeo}

We recall some necessary details from hyperbolic geometry.
The notation is consistent with \textcite{Huang2020a}.  

\begin{definition}\label{def:arc}
        Let $z_1, z_2 \in \C_{\Im \geq 0} \coloneqq \{z \in \C \, | \, \Im(z) \geq 0\}$, the
        upper half complex plane.  We define the following sets:
        \begin{enumerate}
                \item 
                        $\text{Circ}\,({z_1},\, {z_2})$ is the circle through $z_1$ and $z_2$ with
                        centre on the real axis.  If $\Re(z_1) = \Re(z_2)$, this is
                        defined as the infinite line passing through $z_1$ and $z_2$.
                \item $\text{Arc}_{\min}\,(z_1,\, z_2)$ is the arc of
                        $\text{Circ}\,({z_1},\, {z_2})$ in $\C_{\Im \geq 0}$.  If $\Re(z_1) =
                        \Re(z_2)$, then $\text{Arc}_{\min}\,(z_1,\, z_2)$ is $[z_1,
                        z_2]$.
                \item Given $z_1, \ldots, z_m \in \C_{\Im \geq 0}$, the \emph{arc-edge polygon}
                        is defined by: $\text{Poly}\,(z_1) \coloneqq \{z_1\}$ and
                        $\text{Poly}\,(z_1, \ldots, z_m)$ is the smallest simply
                        connected set containing $S$, where
                        \begin{IEEEeqnarray*}{+rCl+x*}
                                S &=& \bigcup_{i, j = 1\ldots m} \text{Arc}_{\min}\,
                                (z_i, z_j).&\qedhere
                        \end{IEEEeqnarray*}
        \end{enumerate}
\end{definition}
Note that, as $\Poly{z_1, \ldots, z_{m-1}} \subseteq \Poly{z_1, \ldots, z_{m-1},
z_m} \subseteq \C_{\Im \geq 0}$, the set $\Poly{Z}$, where $Z$ is a countably infinite sequence of
points in $\C_{\Im \geq 0}$, is well defined as the limit $\lim_{m \to \infty} \Poly{Z_m}$, where $Z_m$
is the length $m$ truncation of $Z$ (see \autocite[p. 111]{Rockafellar1998}).

Definition~\ref{def:arc} forms the basis of the Poincar\'e half plane
model of hyperbolic geometry.  Under the Beltrami-Klein mapping, $f \circ g$, where
\begin{IEEEeqnarray*}{rCl}
        f(z) &=& \frac{2z}{1 + |z|^2},\\
        g(z) &=& \frac{z - j}{z + j},
\end{IEEEeqnarray*} $\C_{\Im \geq 0}$ is mapped
onto the unit disc, and $\Arc{z_1, z_2}$ is mapped to a straight line segment.
We make the following definitions of convexity and the convex hull in the Poincar\'e half plane model.

\begin{definition}
        A set $S \subseteq \C_{\Im \geq 0}$ is called \emph{hyperbolic-convex} or \emph{h-convex} if 
        \begin{IEEEeqnarray*}{rCl}
                z_1, z_2 \in S \implies \Arc{z_1, z_2} \in S.
        \end{IEEEeqnarray*}
        Given a set of points $P \in \C_{\Im \geq 0}$, the \emph{h-convex hull of} $P$ is the
        smallest h-convex set containing $P$.
\end{definition}

Note that h-convexity is equivalent to Euclidean convexity under the Beltrami-Klein
mapping.  $\Arc{z_1, z_2}$ is the minimal geodesic between $z_1$ and $z_2$ under
the Poincar\'e metric, so h-convexity may be thought of as geodesic convexity with
respect to this metric.  We recall the following useful lemma of \textcite{Huang2020a}.

\begin{lemma}
        (Lemma 2.1 \autocite{Huang2020a}): Given a sequence of points $Z \in \C_{\Im \geq 0}$, $\Poly{Z}$ is h-convex.
\end{lemma}

In our terminology, given a sequence of points $Z \in \C_{\Im \geq 0}$, $\Poly{Z}$ is the
h-convex hull of $Z$.

\subsection{SRGs of LTI transfer functions}\label{subsec:lti}
Let $g: L_2 \to L_2$ be linear and time invariant, and denote its transfer
function by $G(s)$.
$g$ maps a complex sinusoid $u(t) = a e^{j\omega t}$ to the complex sinusoid
$y(t) = a|G(j \omega)| e^{j\angle G( j \omega) + j\omega t}$.  These signals do not
belong to $L_2$, but are treated as limits of sequences in $L_2$.  Precisely, we define
the points on the SRG corresponding to sinusoidal signals by taking the 
gain and phase to be
\begin{IEEEeqnarray*}{rCl}
        \lim_{T \to \infty} \frac{\norm{P_T y}}{\norm{P_T u}}\\ 
        \lim_{T \to \infty} \angle(P_T u, P_T y).
\end{IEEEeqnarray*}
Both these limits exist when $u$ and $y$ are sinusoidal. 
The Nyquist diagram $\nyq{G}$ of an operator $g: L_2(\C) \to L_2(\C)$ is the locus of points
$\{G(j \omega)\,|\, \omega \in \R\}$.

\begin{theorem}\label{prop:nyq}
        Let $g: L_2(\C) \to L_2(\C)$ be linear and time invariant, with transfer
        function $G(s)$.  Then
        $\srg{g} \cap \C_{\Im \geq 0}$ is the h-convex hull of $\nyq{G} \cap \C_{\Im \geq 0}$.
\end{theorem}

The proof of Theorem~\ref{prop:nyq} is closely related to the proof of \textcite[Thm.
3.1]{Huang2020a}, and may be found in Appendix~\ref{app:nyq_proof}.  A consequence of
Theorem~\ref{prop:nyq} is that the SRG of an LTI operator is bounded by its
Nyquist diagram. For example, the SRG of the transfer function $1/(s^3+ 5s^2 + 2s + 1)$ is illustrated in
        \ref{fig:third_order}.  Further examples are given in
        \autocite{Chaffey2021b}.

Given Theorem~\ref{prop:nyq}, we recover two familiar properties of
the Nyquist diagram as special cases of Proposition~\ref{prop:properties}, namely
that passivity is equivalent to the Nyquist diagram lying in the right half plane,
and the $H_\infty$ gain is the maximum magnitude of the Nyquist diagram.

        \begin{figure}[h!]
                \centering
                \includegraphics{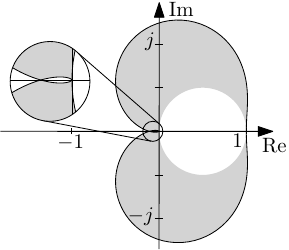}
                \caption{SRG of the transfer function $1/(s^3 + 5s^2 + 2s + 1)$. The
                        black curve is its Nyquist diagram, the grey region is the SRG.}%
                \label{fig:third_order}
        \end{figure}

\section{Scaled relative graphs of static nonlinearities}\label{sec:static}

LTI systems map complex sinusoids to complex sinusoids, and the behavior of an LTI
system on $L_2$ can be fully characterized by its behavior on complex sinusoids.

Similarly, static nonlinearities map square waves to square waves.
Here, we
show that the behavior of single input, single output static nonlinearities on $L_2$, insofar as it is captured by the scaled relative graph, is
fully characterized by their behavior on a two-dimensional subspace of $L_2$ spanned
by two Haar wavelets (truncations of a square wave to a single period).  In
particular, we show that the SRGs of the saturation and ReLU are identical, and closely
related to the SRG of a first order lag.  The use of square waves allows us to test
the effect of different input amplitudes on the output, which is analogous to the use
of sinusoids to test the effect of different input frequencies on the output of an
LTI system.

\begin{proposition}\label{prop:static_bound}
        Suppose $S: L_2 \to L_2$ is the operator given by a SISO static nonlinearity
        $s: \R \to \R$, such that for all $u_1, u_2 \in \R$, $y_i \in s(u_i)$, 
        \begin{IEEEeqnarray}{rCl}
        \mu (u_1 - u_2)^2 \leq (y_1 - y_2)(u_1 - u_2) \leq \lambda 
        (u_1 - u_2)^2.\label{eq:incremental_sector} 
\end{IEEEeqnarray}
        Then the SRG of $S$ is contained within the disc centred at 
        $(\lambda + \mu)/2$ with radius $(\lambda - \mu)/2$.        
\end{proposition}

For a static nonlinearity obeying Condition~\eqref{eq:incremental_sector}, we say
that it is 
\emph{incrementally in the sector $[\mu, \lambda]$}.

\begin{proof}
        Define an operator $\bar S$ by $u \mapsto \bar{y} \coloneqq S(u) - \mu u$.
        Let $\Delta u(t) = u_1(t) - u_2(t)$ and $\Delta \bar{y}(t) = \bar{y_1}(t) -
        \bar{y_2}(t)$.  We drop the $t$ dependence in the remainder of this proof.
        By assumption on $s$, for all $\Delta u$ and corresponding incremental 
        output $\Delta \bar{y}$, we have
        \begin{IEEEeqnarray}{rCl}
                0 &\leq& \Delta u(\Delta y - \mu \Delta u) \leq (\lambda - \mu)
                \Delta u^2,\\
                0 &\leq& \Delta u\Delta \bar{y} \leq (\lambda - \mu)
                \Delta u^2.
        \end{IEEEeqnarray}
        It then follows that $\Delta u \Delta \bar{y} \geq 0$ and $\Delta u \Delta \bar{y} - 
        (\lambda - \mu)\Delta u^2 \leq 0$,
        from which the following series of equivalent statements follow:
        \begin{IEEEeqnarray*}{rCl}
                \Delta u \Delta \bar{y} (\Delta u \Delta \bar{y} -
                (\lambda - \mu)\Delta u^2) &\leq& 0\\
                \Delta u^2 (\Delta \bar{y}^2 - (\lambda - \mu)\Delta u \Delta
                \bar{y}) &\leq& 0\\
                \Delta \bar{y}^2 &\leq& (\lambda - \mu)\Delta u \Delta \bar{y}\\
                \Delta u \Delta \bar{y} &\geq& \frac{1}{\lambda - \mu} \Delta
                \bar{y}^2.
        \end{IEEEeqnarray*}
        This shows that $\bar{S}$ is output-strict incrementally positive with
        constant $1/(\lambda - \mu)$, so its SRG is the disc with centre $(\lambda -
        \mu)/2$ and radius $(\lambda - \mu)/2$.  The result then follows by
        noting that $S$ is the parallel interconnection of $\bar{S}$ with $\mu I$, so
        its SRG is the SRG of $\bar{S}$ shifted to the right by $\mu$.
\end{proof}

The same bounding region can be obtained for the SG with respect to an input
$u^\star$, by restricting the second input in the proof of
Proposition~\ref{prop:static_bound} to be $u^\star$. This is stated formally below.

\begin{proposition}\label{prop:static_bound_SG}
Suppose $S: L_2 \to L_2$ is the operator given by a SISO static nonlinearity $s: \R
\to \R$, such that, for all $u_1 \in \R$, $y_1 \in s(u_1)$, $y^\star \in
s(u^\star)$, 
\begin{IEEEeqnarray}{rCl}
        \mu (u_1 - u^\star)^2 \leq (y_1 - y^\star)(u_1 - u^\star) \leq \lambda 
        (u_1 - u^\star)^2.\label{eq:sector} 
\end{IEEEeqnarray}
        Then the SG of $S$ with resepct to $u^\star$ is contained within the disc centred at 
        $(\lambda + \mu)/2$ with radius $(\lambda - \mu)/2$.        
\end{proposition}

The discs obtained in the previous two propositions are closely related to the
discs of the classical incremental circle criterion \autocite{Zames1968} - indeed,
taking the negative and inverting transforms one to the other.

We now show that, for a large class of systems, the disc bound on the SRG cannot be
improved.  If the characteristic curve of $s$ contains a ``maximal elbow'', that is,
a point where the slope switches from maximum to minimum, then small signals centred around the elbow
can be used to generate the perimeter of the bound of
Proposition~\ref{prop:static_bound}.  Furthermore, if the region of minimum slope
extends to infinity, then large signals can be used to generate the interior of the bound of
Proposition~\ref{prop:static_bound}.  This is formalized in the following two
propositions.  We treat only an elbow from slope $1$ to slope $0$, as a loop
transformation can be used to convert any other elbow to this form.

\begin{proposition}\label{thm:static_nonlinearity}
        Suppose $S:L_2 \to L_2$ is a memoryless nonlinearity defined by a
        map $s:\R \to \R$ which satisfies \eqref{eq:incremental_sector} with $\mu =
        0$ and $\lambda = 1$. Furthermore,
        suppose there are real numbers $u^\star$ and $\delta > 0$, such that, 
        \begin{IEEEeqnarray}{rCl}
                \label{eq:min_slope}
                s(u^\star + \epsilon_u) - s(u^\star) = 0 \mbox{ for all } \epsilon_u \in
        [0, \delta]\\ 
        \label{eq:max_slope}
        s(u^\star) - s(u^\star - \epsilon_l)
        = \epsilon_l\mbox{ for all } \epsilon_l \in [0, \delta].
        \end{IEEEeqnarray}
  Then the SRG of $S$ contains the circle centred at $1/2$ with radius
        $1/2$.
\end{proposition}

\begin{proof}
        We consider two input signals, supported on $[0, 1]$:
        \begin{IEEEeqnarray*}{C}
                u_1(t) = u^\star, \quad
                u_2(t) = \begin{cases}
                                u^\star + \epsilon & 0 \leq t < \tau\\
                                u^\star - \epsilon & \tau \leq t \leq 1,
                        \end{cases}
        \end{IEEEeqnarray*}
        where $\tau \in [0, 1]$.  The corresponding output signals are given by
        \begin{IEEEeqnarray*}{C}
                y_1(t) = s(u^\star), \quad
                y_2(t) = \begin{cases}
                        s(u^\star + \epsilon) & 0 \leq t < \tau\\
                        s(u^\star - \epsilon) & \tau \leq t \leq 1,
                        \end{cases}
        \end{IEEEeqnarray*}
        giving the incremental signals
        \begin{IEEEeqnarray*}{C}
                \Delta u(t) = \begin{cases}
                        -\epsilon & 0 \leq t < \tau\\
                        \epsilon & \tau \leq t \leq 1,
                        \end{cases}\quad
                \Delta y(t) = \begin{cases}
                        0 & 0 \leq t < \tau\\
                        \epsilon & \tau \leq t \leq 1.
                        \end{cases}
        \end{IEEEeqnarray*}
        $\Delta y$ can be written as $k(t) \Delta u(t)$, where
        \begin{IEEEeqnarray*}{rCl}
                k(t) &=& \begin{cases}
                        0 & 0 \leq t < \tau\\
                        1 & \tau \leq t \leq 1.
                        \end{cases}
        \end{IEEEeqnarray*}
        Calculating gain then gives
        \begin{IEEEeqnarray*}{rCl}
                \norm{\Delta y} &=& \left(\int^1_0\hspace{-2mm} k^2(t) \Delta u^2(t)
        \dd{t}\right)^{\frac{1}{2}}
                                = \left(\int^1_\tau \Delta u^2(t)
        \dd{t}\right)^{\frac{1}{2}} 
                                = \gamma \norm{\Delta u},
        \end{IEEEeqnarray*}
        for some $\gamma$ which varies between $0$ and $1$ as $\tau$ varies between
        $1$ and $0$.  It follows that
        \begin{IEEEeqnarray*}{rCl}
                \frac{\norm{\Delta y}}{\norm{\Delta u}} &=& \gamma.
        \end{IEEEeqnarray*}
        Calculating the phase gives
        \begin{IEEEeqnarray*}{rCl}
               \acos\frac{\bra{\Delta u}\ket{\Delta y}}{\norm{\Delta u}\norm{\Delta y}}
                &=& \acos\frac{\int_0^1 k(t) \Delta u^2(t) \dd{t}}{\gamma \norm{\Delta
                u}^2}\\
                &=& \acos\frac{\int_\tau^1 \Delta u^2(t)\dd{t}}{\gamma \norm{\Delta
                u}^2}\\
                &=& \acos(\gamma).
        \end{IEEEeqnarray*}
        Since $\gamma \in [0, 1]$, we can define $\theta$ by $\cos(\theta) = \gamma$.
        We then have the locus of points on the SRG given by
        \begin{IEEEeqnarray*}{rCl}
                \cos(\theta) \exp(\pm j \theta), \quad 0 \leq \theta \leq \pi/2, 
        \end{IEEEeqnarray*}
        which is the circle with centre $1/2$ and radius $1/2$.
\end{proof}

\begin{proposition}\label{prop:static_nonlinearity_2}
        Suppose $S:L_2 \to L_2$ is a memoryless nonlinearity defined by a
        map $s:\R \to \R$ which satisfies \eqref{eq:incremental_sector} with $\mu =
        0$ and $\lambda = 1$, and which satisfies $s(0) = 0$. Furthermore,
        suppose there is a real number $u^\star$ such that 
        \begin{IEEEeqnarray}{rCl}
                s(u^\star + M) - s(u^\star) = 0 \mbox{ for all } M \geq 0\\ 
        s(u^\star) > 0.
        \end{IEEEeqnarray}
        Then the SRG of $S$ is the disc centred at $s(u^\star)/2u^\star$ with radius
        $s(u^\star)/2u^\star$.
\end{proposition}

\begin{proof}
        We consider two input signals, supported on $[0, 1]$:
        \begin{IEEEeqnarray*}{C}
                u_1(t) = M, \quad
                u_2(t) = \begin{cases}
                                M + u^\star & 0 \leq t < \tau\\
                                0 & \tau \leq t \leq 1,
                        \end{cases}
        \end{IEEEeqnarray*}
        where $\tau \in [0, 1]$, and $M \geq u^\star$.  Performing the same
        calculations as in the proof of Proposition~\ref{thm:static_nonlinearity},
        and defining        $\beta(M) \coloneqq s(u^\star)/M$, we have        
\begin{IEEEeqnarray*}{rCl}
                \frac{\norm{\Delta y}}{\norm{\Delta u}} &=& \beta(M)\gamma,\\
                \acos\frac{\bra{\Delta u}\ket{\Delta y}}{\norm{\Delta u}\norm{\Delta y}}
                                                             &=&  \acos(\gamma).
        \end{IEEEeqnarray*}
        Since $\gamma \in [0, 1]$, we can define $\theta$ by $\cos(\theta) = \gamma$.
        We then have the locus of points on the SRG given by
        \begin{IEEEeqnarray*}{rCl}
                \beta(M)\cos(\theta) \exp(\pm j \theta), \quad 0 \leq \theta \leq \pi/2. 
        \end{IEEEeqnarray*}
        This is the circle with centre $\beta(M)/2$ and radius
        $\beta(M)/2$.  Varying $M$ between $u^\star$ and $\infty$ varies $\beta(M)$
        between $s(u^\star)/u^\star$ and $0$, so we fill the disc  with centre
        $s(u^\star)/2u^\star$ and radius $s(u^\star)/2u^\star$.
\end{proof}

Proposition~\ref{prop:static_nonlinearity_2} allows us to gives an exact characterization of the SRGs of
a range of static nonlinearities, including the unit saturation, the ReLU, and the
limiting cases of the relay and ideal diode.

The proof of Proposition~\ref{thm:static_nonlinearity} uses probing signals which
have an arbitrarily small magnitude variation about a ``worst case'' input value.
This shows the local or worst case nature of the SRG --- the boundary of the SRG is generated by these probing
signals. 

\begin{remark}\label{rem:static}
We conclude this section by remarking that the characterization of output-strict
incrementally passive \emph{static} nonlinearities allows the SRGs of a large class
of \emph{dynamic} output-strict incrementally passive nonlinear systems to be
characterized.  Output-strict incremental passivity is preserved under negative
feedback with an incrementally passive system, as shown in
Theorem~\ref{thm:weak_passivity}.  This means
that any scalar system of the form
\begin{IEEEeqnarray*}{rCl}
        \dot y &=& f(u - y),
\end{IEEEeqnarray*}
where $f$ is incrementally in a sector with positive constants, is output-strict
incrementally passive.
\end{remark}

\section{Example 1: feedback with saturation and delay}\label{sec:delay}

In this section, we use SRGs to analyze feedback systems with delays, dynamic
components and static nonlinearities.  We will derive incremental stability bounds which depend both on the delay time and the dynamic time constant, similar to
the state of the art non-incremental bounds obtained using the roll-off IQC
\autocite{Summers2013}.  These bounds are obtained by approximating the SRG of the
delay and the dynamics, treated as a
single component.  In the simple example of this section,
where the dynamic component is LTI,
this approach reduces to the incremental circle criterion. 
However, the approach allows for arbitrary dynamic components,
as shown in Section~\ref{sec:congestion}. One of the advantages of our approach is
the derivation of stability margins and incremental $L_2$ gain bounds for the closed
loop.

We begin with the system of Figure~\ref{fig:delay_1}, showing a time delay and an LTI transfer function in feedback with a
$1/\beta$-output-strict incrementally passive component $\Delta$.

\begin{figure}[ht]
        \centering
        \includegraphics{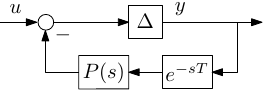}
        \caption{Simple system with delay in the feedback loop.}%
        \label{fig:delay_1}
\end{figure}

We take $P(s) = s^2/(s^3 + 2s^2 + 2s + 1)$, also
considered in \autocite[\S 3]{Megretski1997}.  
The Nyquist diagram of $P(s)$ cascaded with the delay, and a bounding
approximation of the SRG, are shown in the left hand side of
Figure~\ref{fig:delay_2}.  As the delay is increased, the Nyquist diagram, and hence
the SRG, extend further into the left half plane.

\begin{figure}[ht]
        \centering
        \includegraphics{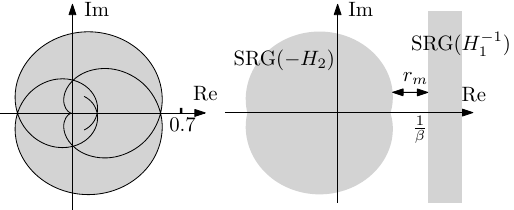}
        \caption{Left: Nyquist diagram of $e^{-sT}s^2/(s^3 + 2s^2 + 2s + 1)$ (black) and a bounding
        approximation of its SRG.  Right: feedback with $1/\beta$-output-strict
incrementally passive system.}%
        \label{fig:delay_2}
\end{figure}

Applying Theorem~\ref{thm:robust} with $H_2 = e^{-sT}P(s)$ and $H_1 = \Delta$, we obtain the right hand side of Figure~\ref{fig:delay_2}.  Stability is
verified if the delay SRG always has real part
greater than $1/\beta$, which ensures that $r_m > 0$.
Solving numerically for $\min_\omega \Re(P(j\omega) e^{j\omega T})$ gives a stability bound on
$\beta$, as a function of $T$, shown in
Figure~\ref{fig:bound}, which also shows the non-incremental stability bound obtained
by \textcite{Megretski1997} using IQC analysis, for the particular case where
$\Delta$ is a saturation.  For short delay times, the
non-incremental bound is shown to tend to infinity, using the Zames-Falb-O'Shea
multiplier.  The incremental bound obtained using SRG analysis has a non-smooth point
where the leftmost segment of the Nyquist diagram switches, and is bounded for all
delay times.

The SRG analysis gives a bound which guarantees finite incremental $L_2$ gain, a 
stronger property than the $L_2$ gain from IQC analysis.  Finite incremental
$L_2$ gain in particular implies input-output Lipschitz continuity. To the best of the authors'
knowledge (and as also noted in \autocite{Wang2019}),
\autocite{Jonsson2001} is the only application of incremental IQCs to stability
analysis of feedback systems, with only a very weak form of stability
guaranteed.  As noted by
\textcite{Kulkarni2001}, stability results using Zames-Falb-O'Shea
and Popov multipliers do not guarantee continuity, as these multipliers do not
preserve the incremental passivity of static nonlinear elements.  The situation for
proving finite incremental $L_2$ gain with these multipliers is similar; the loss of
incremental passivity of the static nonlinearity means the incremental passivity
theorem cannot be applied, so another method of proving stability is needed.  One such method would be to
apply Theorem~\ref{thm:robust} to the transformed loop, and indeed there are
multipliers which destroy incremental passivity but which still verify an incremental
$L_2$ gain bound.  For this particular example, the transfer function $(s +
1)/(s - 1)$ could be used as a multiplier, although it gives a more conservative
bound than Figure~\ref{fig:bound}. 
Global \autocite{Wang2019}, universal
\autocite{Manchester2018} and equilibrium-independent \autocite{Hines2011} $L_2$ gain
are weaker than incremental $L_2$ gain but stronger than $L_2$ gain, and afford
differing levels of tractability.

In addition to proving incremental $L_2$ stability, we can give an incremental $L_2$
gain bound. For a fixed $\beta$, $1/r_m$ is an incremental $L_2$ gain bound from $u$
to $y$, which depends on
the time delay $T$.  For $\beta = 1$, this bound is plotted in Figure~\ref{fig:bound}. 

\begin{figure}[ht]
        \centering
        \includegraphics[width=\linewidth]{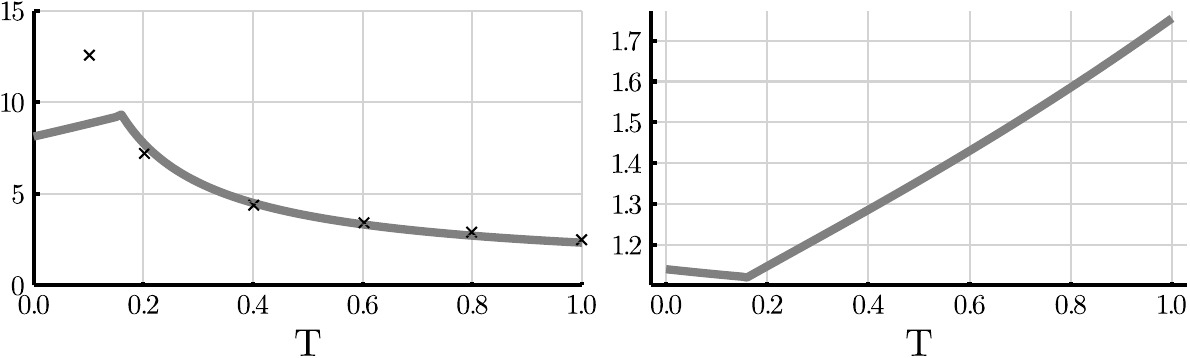}
        \caption{Left: The grey line is an upper bound on $\beta$ which guarantees that the system of Figure~\ref{fig:delay_1}
                has bounded incremental $L_2$ gain with a delay of $T$. The crosses
                give a bound on $\beta$ which guarantees (non-incremental) $L_2$ stability,
obtained using IQC analysis \autocite[Fig. 6]{Megretski1997}. Right: the incremental $L_2$ gain bound from $u$ to $y$ for $\beta
                = 1$.  }%
        \label{fig:bound}
\end{figure}

The motivation behind the traditional structure of the Lur'e system is to put all of
the ``troublesome'' elements in the nonlinear component, and all of the dynamics in
the LTI component.  The availability of explicit SRGs
for elements which are usually troublesome, such as saturations and delays, means
that this structure is not necessarily ideal for SRG analysis, and the feedback
system may be better modelled in a different way.  This is illustrated in the
following two examples.

\section{Example 2: cyclic feedback systems}\label{sec:cyclic}

We now turn to the analysis of cascades. Such systems form the basis of cyclic
feedback systems, which are often found
in biological models \autocite{Tyson1978}, among many other application domains (see,
for example, the discussion of \textcite{Mallet-Paret1996}).
In Theorem~\ref{thm:cascade}, we give the SRG of a cascade of output-strict
incrementally positive systems, which
represents a novel system constraint which cannot be represented as an incremental
IQC.  A
gain margin condition applied to a cascade in unity gain negative feedback gives rise to the incremental secant condition
\autocite{Sontag2006}.  The cascade SRG thus generalizes the incremental secant
condition to arbitrary feedback interconnections and disturbances.
The result we give here is tight in the sense that stronger
conditions than output strict incremental positivity of the plants are required for
any stronger bound.

\begin{figure}[ht]
        \centering
        \includegraphics{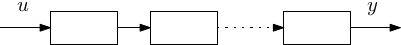}
        \caption{Cascade of $n$ systems.}%
        \label{fig:cascade}
\end{figure}

\begin{theorem}\label{thm:cascade}
Consider the cascade of $n$ output-strict incrementally positive systems, with
parameters $1/\gamma_i$, $i = 1, \ldots, n$, shown in Figure~\ref{fig:cascade}. 
The SRG of the cascade is contained within the region with perimeter
\begin{IEEEeqnarray}{rCl}
        z(\phi) = \gamma_1\gamma_2\ldots\gamma_n \left(\cos\frac{\phi}{n}\right)^n e^{-j \phi}, \quad -\pi
                \leq \phi < \pi. \label{eq:cascade}
\end{IEEEeqnarray}
\end{theorem}

\begin{proof}
The SRG of the $i^\text{th}$ system is the disc with centre $\gamma_i/2$ and radius
$\gamma_i/2$.  The perimeter of this disc has the parameterization 
\begin{IEEEeqnarray}{rCl}
z_i(\theta) &=&
\gamma_i \cos(\theta) e^{-j\theta}\quad -\pi/2 \leq \theta < \pi/2 \label{eq:disc}
\end{IEEEeqnarray}
As this disc satisfies the right hand arc property, the SRG of the full cascade is
the product of $n$ discs.  We claim that
the perimeter of this SRG has the parameterization given by
Equation~\ref{eq:cascade}.

For instance, take any $z_1, z_2, \ldots, z_n$.
Using \eqref{eq:disc} and Proposition~\ref{prop:composition} gives the point
\begin{IEEEeqnarray}{rCl}
        w = \gamma_1\ldots\gamma_n
        \cos(\theta_1)\ldots\cos(\theta_n) e^{-j (\theta_1 + 
        \ldots + \theta_n)}, \label{eq:w}
\end{IEEEeqnarray}
for $-\pi < \theta_1, \theta_2, \ldots, \theta_n < \pi$.  Letting $\theta_1 =
\theta_2 = \ldots = \theta$, and setting $\phi = n\theta$ gives the parameterization
\eqref{eq:cascade} (noting that $-\pi \leq \phi < \pi$ as \eqref{eq:cascade} is
$2\pi$-periodic).  This shows that all the points $z(\phi)$ lie within the SRG.  To
show that they are indeed on the perimeter of the SRG, we take any point $w$ and show
that its magnitude is smaller than the point $z(\phi)$ with the same argument.  This
follows from \eqref{eq:w} if we can show that 
\begin{IEEEeqnarray*}{rCl}
        \cos(\theta_1)\cos(\theta_2)\ldots\cos(\theta_n) \leq \cos(\theta_1 +
        \theta_2 + \ldots + \theta_n).
\end{IEEEeqnarray*}
This is proved in \autocite{Sontag2006}: $f(\phi) = -\ln\cos(\phi)$ is convex on
$(-\pi/2, \pi/2)$.  Applying Jensen's inequality gives $f(\sum_i \theta_i) \leq
\sum_i f(\theta_i)$, and the required inequality follows by taking the exponential.
Note that the inequality still holds in the limit as one angle $\theta_i \to \pm
\pi/2$.
\end{proof}

The SRG given by Theorem~\ref{thm:cascade} is illustrated in
Figure~\ref{fig:cascade_srg}.
For $n = 2$, this SRG is a special case of \autocite[Thm. 2]{Huang2020}.  For $n >
2$, this SRG is a novel result.
The intercept with the negative real axis is at the
point $z(\pi) = -\gamma_1\gamma_2\ldots\gamma_n \left(\cos\frac{\pi}{n}\right)^n$.  A direct
application of Theorem~\ref{thm:Nyquist} to a cascade in unity gain negative feedback thus gives the following incremental secant
condition.

\begin{corollary}\label{cor:secant}
        Suppose the system of Figure~\ref{fig:cascade} is placed in unity gain
        negative feedback, where the $n$ interconnected systems
        are each output-strict incrementally positive with parameters $\gamma_i$, $i =
        1, \ldots, n$. The feedback interconnection has a finite incremental $L_2$ gain if
        \begin{IEEEeqnarray*}{rCl}
                \gamma_1\gamma_2\ldots\gamma_n < \left(\sec\frac{\pi}{n}\right)^n.
        \end{IEEEeqnarray*}
\end{corollary}

To see that the cascade SRG expresses a more general constraint than possible with an
incremental IQC, we can take the $n=2$ case of
Equation~\eqref{eq:cascade}, and eliminate the parameter $\phi$.  This gives the
following equality constraint on the boundary of the SRG:
\begin{IEEEeqnarray*}{rCl}
\bra{u_1 - u_2}\ket{y_1 - y_2} + \norm{u_1 - u_2}\norm{y_1 - y_2} - 2\norm{y_1 -
y_2}^2.
\end{IEEEeqnarray*}
The middle term cannot be expressed as an incremental IQC.

The cascade SRG allows several other useful values to
be computed. An incremental $L_2$ gain bound can be found by minimizing the distance between $-1$
and $ 1/(\gamma_1\ldots\gamma_n
        \cos(\theta_1)\ldots\cos(\theta_n) e^{-j (\theta_1 + 
        \ldots + \theta_n)})$.  This distance is shown for $n
        = 4$ in Figure~\ref{fig:cascade_inverse}.
Furthermore,  we can
calculate the shortage of input-strict incremental positivity of the cascade by finding the distance
the SRG extends into the left half plane.  For example, for a cascade of two systems,
 the shortage of input-strict incremental positivity is
 $\gamma_1\gamma_2/8$ \autocite[p. 118]{Lawrence1972}. Stan \emph{et al.} \autocite{Stan2007} show that if the coupling strength in a network of oscillators modelled as
cascade feedback systems is large enough compared to the shortage of each oscillator,
the network will synchronize.

\begin{figure}[ht]
        \centering
        \includegraphics{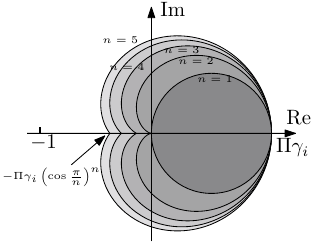}
        \caption{SRGs of the cascade of Figure~\ref{fig:cascade}, where 
        subsystem $i$ is $\gamma_i$-output-strict incrementally positive, for 1 to 5 subsystems.}%
        \label{fig:cascade_srg}
\end{figure}

\begin{figure}[ht]
        \centering
        \includegraphics{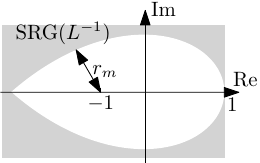}
        \caption{Inverse SRG of a cascade of four output-strict incrementally
                positive 
        systems.  The stability margin is $r_m$. The intercept with the negative real
axis is at $-1/(\Pi_i \gamma_i (\cos(\pi/n))^n)$.}%
        \label{fig:cascade_inverse}
\end{figure}

SRG analysis allows the incremental secant condition to be generalized beyond
negative feedback interconnections.  For example,
if an uncertain gain $k_\Delta$ is placed in feedback with the cascade, as shown
in Figure~\ref{fig:cascade_ff}, we can give a bound on $k_\Delta$
for which incremental stability is guaranteed.  The inverse SRG of the cascade
(Figure~\ref{fig:cascade_inverse}) is shifted to the left by $k_\Delta$; if it does
not intersect $-1$, the closed loop has finite incremental gain.  This allows us to conclude stability if
\begin{IEEEeqnarray*}{rCl}
        \left(\gamma_1\gamma_2\ldots\gamma_n\right)^{-1} > k_\Delta > 1 -
        \left(\gamma_1\gamma_2\ldots\gamma_n\left(\cos\frac{\pi}{n}\right)^n\right)^{-1}.
\end{IEEEeqnarray*}

\begin{figure}[ht]
        \centering
        \includegraphics{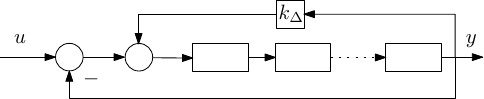}
        \caption{Cascade of subsystems with an
        uncertain feedback gain $k_\Delta$, in unity gain negative feedback.}%
        \label{fig:cascade_ff}
\end{figure}

\begin{remark}
        Remark~\ref{rem:static} showed how the bounding SRG for a static nonlinearity
        that is incrementally in a sector could be used to determine bounding SRGs
        for dynamic nonlinearities.  The cascade SRG derived in this section allows
        us to extend this idea to dynamic nonlinearities described by differential
        equations of the form
        \begin{IEEEeqnarray}{rCl}
                \dot{y} &=& f(g(u) - y).\label{eq:dynamic_nonlinearity}
        \end{IEEEeqnarray}
        Suppose that $f$ is incrementally in the sector $[\mu_1, \gamma_1]$ (in the
        sense of Proposition~\ref{prop:static_bound}), and $g$
        is incrementally in the sector $[\mu_2, \gamma_2]$.  For simplicity, assume
        $\mu_1 = \mu_2 = 0$.  The system of Equation~\eqref{eq:dynamic_nonlinearity}
        can be represented as $f$ in negative feedback with an integrator, with the
        nonlinearity $g$ at the input.  It follows from Remark~\ref{rem:static} and
        the Theorem~\ref{thm:cascade} that this system has a bounding SRG given by
        the $n=2$ case of Figure~\ref{fig:cascade_srg}.
\end{remark}

\section{Example 3: combining cascades and delays}\label{sec:congestion}

In this final example, we combine a delay with a cascade of two output-strict incrementally
positive systems, and revisit the Internet
congestion control example of \textcite{Summers2013}. In that paper,
equilibrium-independent IQCs are verified numerically in order to compute a bound on
the variables $\delta$, $\beta$ and $N_u$ in the left of Figure~\ref{fig:congestion}, which
guarantees (non-incremental) input/output stability of the system.  Here, we derive a
bound which guarantees finite incremental gain of the system.

\begin{figure}[ht]
        \centering
        \includegraphics[width=\linewidth]{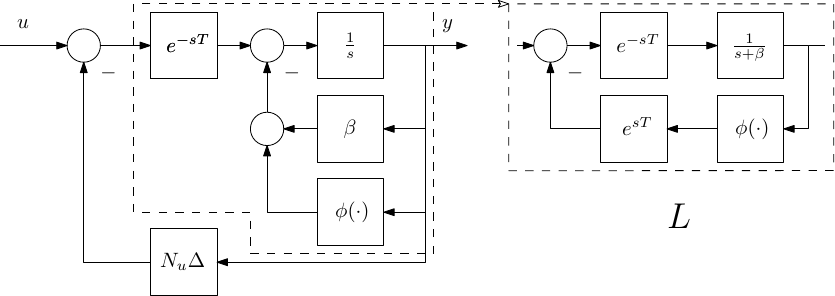}
        \caption{Left: Internet congestion control example of \autocite{Summers2013}.
        $\beta > 0$, $\phi(w)$ is $1/\gamma$-output-strict incrementally positive, $N_u
\in \mathbb{N}$, $\Delta$ is $\delta$-output-strict incrementally positive. $0 < \gamma
< \beta$.  Right:
equivalent representation of the forward path, $L$.}%
        \label{fig:congestion}
\end{figure}

In order to combine the delay and first order lag, we rearrange the forward path as shown
in the right of Figure~\ref{fig:congestion}.  This is the negative feedback
interconnection of $H_1 = e^{-sT}/(s + \beta)$ and $H_2 = -e^{sT}\phi(\cdot)$.
Bounding SRGs for $H_1^{-1}$ and $-H_2$, illustrated in the top of
Figure~\ref{fig:congestion_forward_SRG}, are derived using Theorem~\ref{prop:nyq},
Proposition~\ref{prop:properties} and Proposition~\ref{prop:composition}.    Combining these gives a delay-dependent bounding SRG,
shown in the bottom left of Figure~\ref{fig:congestion_forward_SRG}, 
for the forward path.

\begin{figure}[ht]
        \centering
        \includegraphics{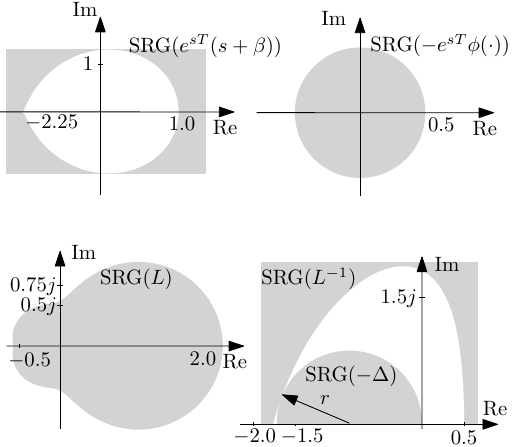}
        \caption{Top: bounding SRGs for $e^{sT}(s+\beta)$ (left) and
                $-e^{sT}\phi(\cdot)$ (right),  for $\beta = 1$, $T = 1$, $\gamma =
        0.5$.  Adding these and inverting gives a bounding SRG for $L$, shown on the
        bottom left.  Bottom right: inverse of $L$ and negative of an
$r$-output-strict incrementally positive block $\Delta$. Only the upper half
is shown.}%
        \label{fig:congestion_forward_SRG}
\end{figure}

To apply Theorem~\ref{thm:robust}, we solve for the largest radius $r$ as shown in
the bottom right of Figure~\ref{fig:congestion_forward_SRG}, before the two SRGs
overlap.  This is equal to the
reciprocal of the distance the SRG of $L$ extends into the left half
plane, which is solved 
numerically. This gives the bound on $N_u/\delta$, plotted in
Figure~\ref{fig:congestion_gain_bound}, that guarantees an incremental $L_2$ gain bound for
the closed loop.

\begin{figure}[ht]
        \centering
        \includegraphics[width=0.7\linewidth]{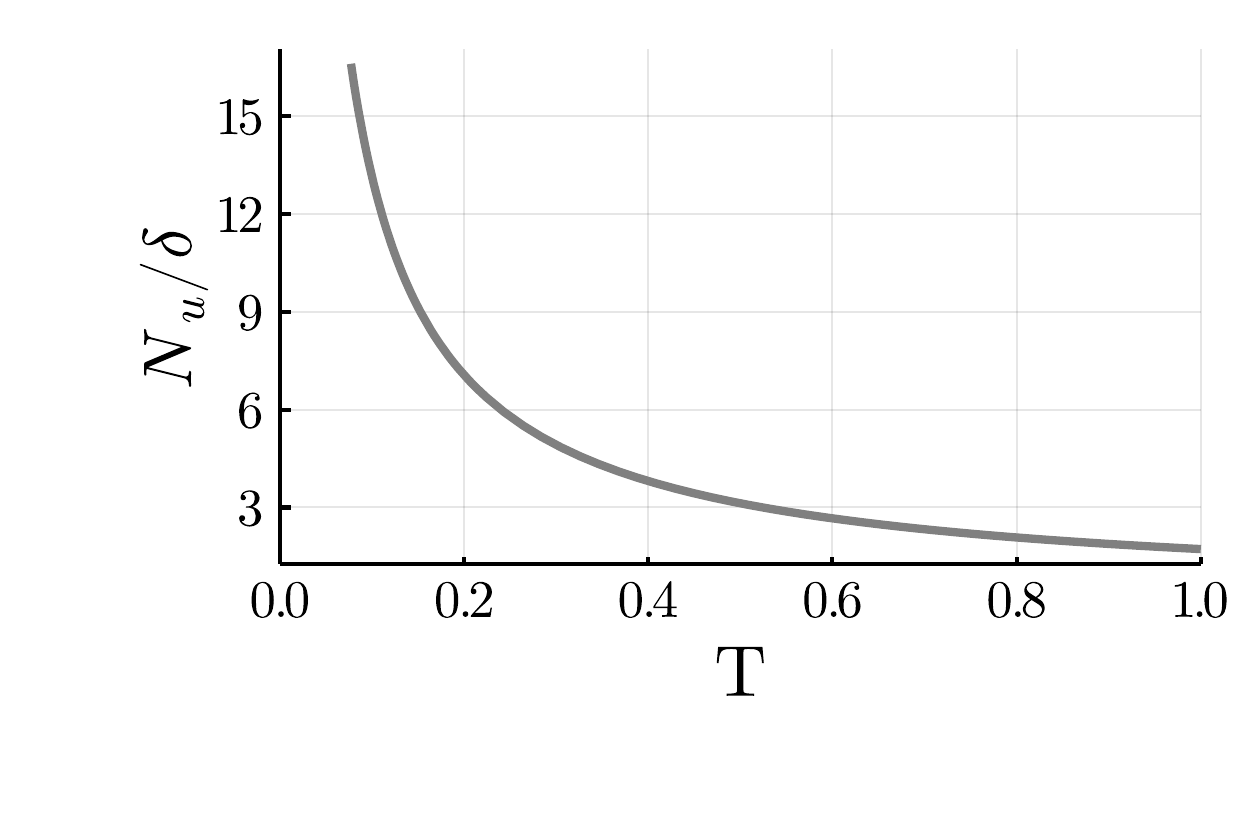}
        \caption{Upper bound on $N_u/\delta$ for the system of
        Figure~\ref{fig:congestion} to have finite incremental $L_2$ gain, derived by
applying Theorem~\ref{thm:robust}. Plotted for $\beta = 1$, $\gamma = 0.5$.}%
        \label{fig:congestion_gain_bound}
\end{figure}

\section{Conclusions}\label{sec:conclusions}

We have introduced the tool of Scaled Relative Graphs to system analysis, and used it
to analyze the incremental stability of operators in feedback.  Characterizing
stability by the separation of two SRGs unifies existing theorems such as the
incremental small gain and passivity theorems, the incremental circle criterion and the
incremental secant condition, using an intuitive graphical language.  This graphical
language is particularly suited to the calculation and visualization of stability
margins, and furthermore allows the input-output gain of a feedback system to be
estimated.  It also allows for a formulation of $H_\infty$ control design for
nonlinear operators. There are many questions for future work; here we will list only a
couple.  

The SRG composition rules rely on a worst-case assumption: that the \emph{same} signals
correspond to the worst-case points in the SRGs of both systems.  When dealing with
interconnections of individual systems (rather than classes of systems), this can be
conservative.  For example, applying Theorem~\ref{thm:robust} to the negative
feedback interconnection of $1/(s+1)$ and $e^{-sT}$ does not give a guarantee of
stability; we know from the Nyquist criterion, however, that unity-gain negative feedback around $e^{-sT}/(s+1)$ does
give a stable system.  A second case where the analysis appears to be conservative is
in the study of multiple input, multiple output systems.  Understanding when SRG analysis is tight, and when it is
conservative, is a topic of ongoing research.  We expect that twenty years of IQC
analysis will contribute to further developing SRG analysis.

A second question is concerned with computation of SRGs.  Efficient algorithms for
computing or approximating the SRGs of nonlinear operators
defined by state space models, or directly from input/output data, are an interesting
topic for future research.

A third question is concerned with the extension of the Nyquist theorem to the
general case of unstable open loop plants.

The SRG characterization of a system can be tightened by taking the intersection of
several bounding SRGs, similar to taking the intersection of several IQCs.  SRG
analysis also allows a characterization to be \emph{loosened} by taking the union of
SRG: for example, a system might \emph{either} be passive, \emph{or} have small gain
(or both).  This is explored in \autocite{Chaffey2022}.

We hope that the graphical analysis presented in this paper will further narrow the
gap between linear and nonlinear control theory.

\appendices
\section{Proof of Theorem \ref{prop:nyq}} \label{app:nyq_proof}

        The proof has three components.  We begin by showing that, for an LTI transfer
        function $G(s)$, the Nyquist
        diagram at the frequencies $n2\pi/T$ is a subset of the SRG of $G(s)$. We
        then show, for operators on the space $L_{2,
        T}$ of $T-$periodic, finite energy signals, the SRG is in the convex hull of the points generated by applying the operator
        to the basis of $L_{2, T}$ given by $\{e^{jt n2\pi/T}\}_{n \in \mathbb{Z}}$,
        which are exactly the points on the Nyquist diagram.  The result then follows by taking
        the limit as $T\to\infty$, analogous to the classical derivation of the
        Fourier transform from the Fourier series.

        We begin by observing that the point on the Nyquist diagram of $G$ corresponding to
        frequency $\omega \in \R$ is precisely $z_G(e^{j\omega t})$. Set 
        $u = a e^{j\omega t}$, then $y = G(u) = \alpha a e^{j\omega t +
        j\psi}$, where $\alpha = |G(j\omega)|$ and $\psi = \angle{G(j \omega)}$. A direct
        calculation gives 
        \begin{IEEEeqnarray*}{rCl}
                \bra{u}\ket{y} &=& \int^T_0 u(t) \bar{y}(t) \mathrm{d} t\\
                               &=& T \alpha a^2 e^{j\phi},\\
                \norm{u} &=& \sqrt{T} a,\\
                \norm{y} &=& \sqrt{T} \alpha a,
        \end{IEEEeqnarray*}
        where $\bra{\cdot}\ket{\cdot}$ is the inner product on $L_{2, T}$.  
        It follows immediately that 
        \begin{IEEEeqnarray*}{rCl}
                z_G(u) &=& \alpha e^{j \psi},
        \end{IEEEeqnarray*}
        that is, the point on the Nyquist diagram of $G$ corresponding to frequency
        $\omega$.

        The next part of the proof closely follows \textcite[Thm 3.1]{Huang2020a}.
        In the interests of brevity, we point out only the main arguments and
        modifications required to that proof.

        Let $G$ be an LTI operator on $L_2$.  The restriction of $G$ to $L_{2, T}$ is
        then an operator on $L_{2, T}$.  Let $\mathcal{B}$ be the set of functions in
        $t$ given by $\mathcal{B} = \{e^{jtn2\pi/T}, n \in \mathbb{Z}\}$.  We show that
        \begin{IEEEeqnarray}{rCl}
                z_G(\text{span}(\mathcal{B})\setminus\{0\}) =
                \Poly{z_G(\mathcal{B})}.\label{eq:countable}
        \end{IEEEeqnarray}

        We begin by noting that $\mathcal{B}$ is an orthonormal basis for $L_{2, T}$, and in particular,
        for all $u, v \in \mathcal{B}$, $u \neq v$,
        $\bra{v}\ket{u} = \bra{v}\ket{Gu} = \bra{Gv}\ket{u} = \bra{Gv}\ket{Gu} = 0$.
        Therefore, the result of Part 2 of the proof of \textcite[Thm.
        3.1]{Huang2020a} holds: for all such $u, v$, we have
        \begin{IEEEeqnarray*}{rCl}
                z_G(\text{span}(u, v)\setminus\{0\}) = \Arc{z_G(u), z_G(v)}.
        \end{IEEEeqnarray*}
        The only modification required to the proof is that the inner product here is
        complex valued, and the real part must be taken.  Parts 3 and 4 of the proof
        of \textcite[Thm. 3.1]{Huang2020a} show that
        $z_G(\text{span}(\mathcal{B})\setminus\{0\}) \subseteq \Poly{z_G(\mathcal{B})}$ and 
        $z_G(\text{span}(\mathcal{B})\setminus\{0\}) \supseteq \Poly{z_G(\mathcal{B})}$
        respectively, with the proof requiring only the additional fact that
        $\Poly{S}$ (in the proof of \autocite[Thm. 3.1]{Huang2020a}) is defined for a countably infinite set, as described in
        Section~\ref{sec:hgeo}.  This concludes the second part of the proof: 
        $z_G(\text{span}(\mathcal{B})\setminus\{0\}) = \Poly{z_G(\mathcal{B})}$.
        
        Finally, we extend to aperiodic signals by letting the period $T \to \infty$ and the
        fundamental frequency $2\pi/T \to 0$.  In the interests of brevity, we give
        the proof here assuming that the Fourier transform of the input $u(t)$ is
        Riemann integrable.  The result can be extended to arbitrary functions on
        $L_2$ using the same machinery for defining the Fourier transform on $L_2$ -
        see, for instance, \textcite[Chap. 9]{Rudin1987}. We first note that $z_G(ae^{i\omega
        t})$ may be computed using the inner product and norm on $L_2$, rather than
        $L_{2, T}$, as a limit, and the result will be unchanged.
        Let $u(t)$ be an input signal on $L_2$, and $y(t)$ the corresponding output.
        The Fourier inversion theorem gives
        \begin{IEEEeqnarray}{rCl}
                y(t) &=& \frac{1}{\sqrt{2\pi}} \int_{-\infty}^{\infty}
                G(j\omega) \hat{u}(\omega) e^{j\omega t}\mathrm{d}\omega.\label{eq:inversion_int}
        \end{IEEEeqnarray}
        Let
        \begin{IEEEeqnarray*}{rCl}
                \frac{\Delta \omega}{\sqrt{2\pi}} \sum_{n = -\infty}^{\infty} G(jn\Delta \omega) \hat{u}(n
                \Delta \omega)e^{jn \Delta \omega} 
        \end{IEEEeqnarray*}
        be a Riemann sum approximation of the right hand side of \eqref{eq:inversion_int}, with uniform
        spacing $\Delta \omega$.  By \eqref{eq:countable}, we know this sum belongs to
        $\Poly{\{G(j \Delta \omega) e^{jn\Delta \omega_0 t}\}_{n \in \mathbb{Z}}}
        \subseteq \Poly{\{G(j \omega) e^{j\omega t}\}_{\omega \in \R}}$. 
        Letting $\Delta \omega \to 0$,  
        we have that the right hand side of \eqref{eq:inversion_int} belongs to
        $\Poly{\{G(j\omega) e^{j\omega t}\}_{\omega \in \R}}$, noting that the restriction of
        the Nyquist diagram to $\C_{\Im \geq 0}$ is compact in $\C$.   
        Note that this is precisely the h-convex hull of the Nyquist diagram of $G$. 
        \qed

\printbibliography

\begin{IEEEbiography}[{\includegraphics[width=1in,height=1.25in,clip,keepaspectratio]{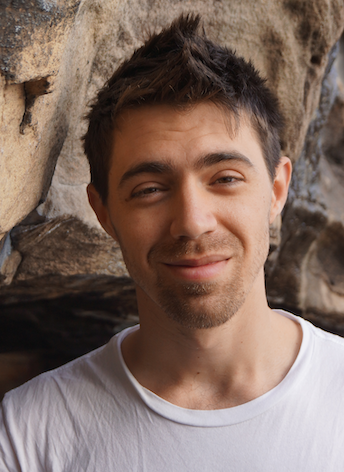}}]{Thomas
        Chaffey} (S17) received the B.Sc. (advmath) degree in
        mathematics and computer science in 2015 and the M.P.E. degree in mechanical
        engineering in 2018, from the University of Sydney, Australia, and the Ph.D.
        degree in engineering from the University of
        Cambridge, UK, in 2022.  He currently holds the Maudslay--Butler Research
        Fellowship at Pembroke College, University of Cambridge.  His research interests are in nonlinear systems, circuits and
        optimization.  He received the Best Student Paper Award at
        the 2021 European Control Conference and the Outstanding Student Paper Award
        at the 2021 IEEE Conference on Decision and Control.
\end{IEEEbiography}

\begin{IEEEbiography}[{\includegraphics[width=1in,height=1.25in,clip,keepaspectratio]{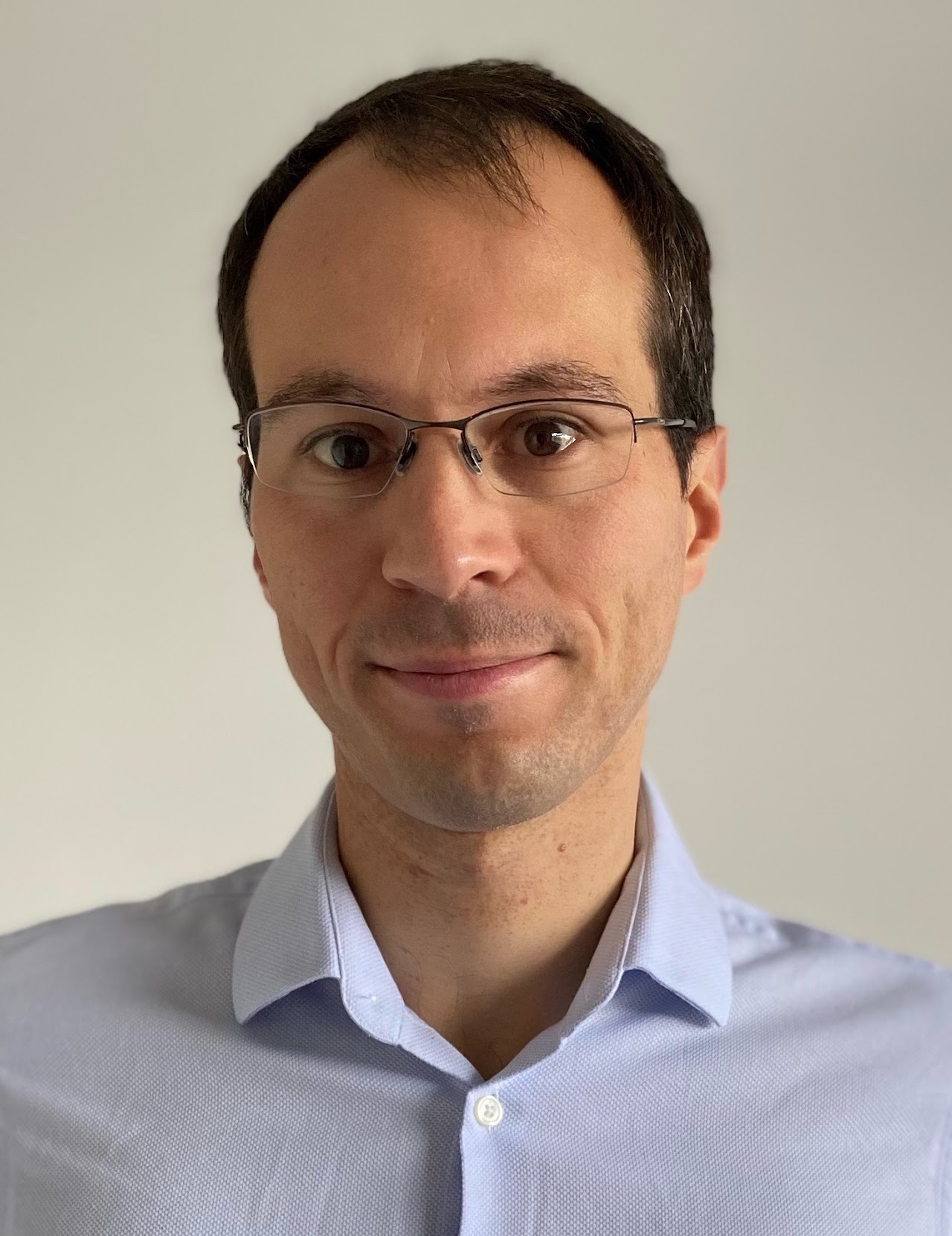}}]{Fulvio Forni}
 received the Ph.D. degree in computer science and control engineering from the University of Rome Tor Vergata, Rome, Italy, in 2010. In 2008–2009, he held visiting positions with the LFCS, University of Edinburgh, U.K. and with the CCDC of the University of California Santa Barbara, USA. In 2011–2015, he held a post-doctoral position with the University of Liege, Belgium (FNRS). He is currently a Lecturer with the Department of Engineering, University of Cambridge, U.K. Dr. Forni was a recipient of the 2020 IEEE CSS George S. Axelby Outstanding Paper Award.
\end{IEEEbiography}

\begin{IEEEbiography}[{\includegraphics[width=1in,height=1.25in,clip,keepaspectratio]{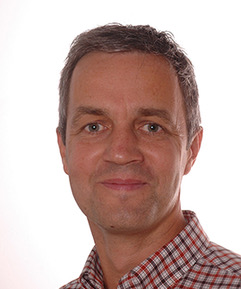}}]{Rodolphe Sepulchre} (M96,SM08,F10)  received the engineering degree and the Ph.D. degree from the Université catholique de Louvain in 1990 and in 1994, respectively.  His is Professor of engineering at Cambridge University since 2013.  His research interests are in nonlinear control and optimization, and more recently neuromorphic control.  He co-authored the monographs "Constructive Nonlinear Control" (Springer-Verlag, 1997) and "Optimization on Matrix Manifolds" (Princeton University Press, 2008). 
He is Editor-in-Chief of IEEE Control Systems. In 2008, he was awarded the IEEE Control Systems Society Antonio Ruberti Young Researcher Prize. He is a fellow of IEEE, IFAC,  and SIAM. He has been IEEE CSS distinguished lecturer  between 2010 and 2015. In 2013, he was elected at the Royal Academy of Belgium.
\end{IEEEbiography}
\end{document}